\documentclass[a4paper,reqno,12pt]{amsart}
\usepackage[margin=2.5cm,papersize={17.75cm,26cm}]{geometry}
\usepackage[all]{xy}           %For commutative diagrams
\usepackage{amssymb}           %For double arrows
\usepackage{hyperref}
\usepackage{eucal,enumitem}
\usepackage{graphicx,color}

\setenumerate[1]{label=(\roman*), ref=(\roman*)}
\numberwithin{equation}{section}
\usepackage[active]{srcltx} %para localizar lo visualizado

\newtheorem{definition}{Definition}[section]
\newtheorem{lemma}[definition]{Lemma}
\newtheorem{theorem}[definition]{Theorem}
\newtheorem{proposition}[definition]{Proposition}
\newtheorem{corollary}[definition]{Corollary}
\newtheorem{remarkth}[definition]{Remark}
\newtheorem{example}[definition]{Example}
\newenvironment{remark}{\begin{remarkth}\upshape}{\hfill$\diamond$\end{remarkth}}

\renewcommand{\emph}[1]{{\bfseries\itshape{#1}}}

\newcommand{\ds}{\displaystyle}

\newcommand{\R}{\mathbb{R}}      %Numeros reales
\newcommand{\lcf}{\lbrack\! \lbrack}
\newcommand{\rcf}{\rbrack\! \rbrack}

\newcommand{\I}{I\mkern-7muI}

\makeatletter
\newcommand\prol{\@ifstar{\@proldf}{\@prolpf}}  %% if * dual else primal
\def\@prolpf{\@ifnextchar[{\@prolpf@wrt}{\@prolpf@}}
\def\@prolpf@wrt[#1]#2{\@ifnextchar[{\@prolpf@wrt@at{#1}{#2}}{\@prolpf@wrt@{#1}{#2}}}
\def\@prolpf@wrt@at#1#2[#3]{\prolsymbol^{#1}_{#3}#2}
\def\@prolpf@wrt@#1#2{\prolsymbol^{#1}#2}
\def\@prolpf@#1{\@ifnextchar[{\@prolpf@at{#1}}{\@prolpf@@{#1}}}
\def\@prolpf@at#1[#2]{\prolsymbol_{#2}#1}
\def\@prolpf@@#1{\prolsymbol#1}
\def\@proldf{\@ifnextchar[{\@proldf@wrt}{\@proldf@}}
\def\@proldf@wrt[#1]#2{\@ifnextchar[{\@proldf@wrt@at{#1}{#2}}{\@proldf@wrt@{#1}{#2}}}
\def\@proldf@wrt@at#1#2[#3]{\prolsymbol^{*#1}_{#3}#2}
\def\@proldf@wrt@#1#2{\prolsymbol^{*#1}#2}
\def\@proldf@#1{\@ifnextchar[{\@proldf@at{#1}}{\@proldf@@{#1}}}
\def\@proldf@at#1[#2]{\prolsymbol^*_{#2}#1}
\def\@proldf@@#1{\prolsymbol^*#1}
\def\prolsymbol{\mathcal{T}}
\makeatother

\def\lcf{\lbrack\! \lbrack}
\def\rcf{\rbrack\! \rbrack}
\usepackage{fancybox}

\begin{document}

\title[Kinematic reduction and the Hamilton-Jacobi equation]{Kinematic reduction and the Hamilton-Jacobi equation}

\author{}.
\author[\tiny Barbero]{Mar\'ia Barbero-Li\~n\'an}
\address{M. Barbero-Li\~n\'an:
Instituto de Ciencias Matem\'aticas, CSIC-UAM-UC3M-UCM, C/Nicol\'as
Cabrera 13-15, 28049 Madrid, Spain} \email{mbarbero@icmat.es}

\author[\tiny de Le\'on]{Manuel de Le\'on}
\address{M. de Le\'on:
Instituto de Ciencias Matem\'aticas, CSIC-UAM-UC3M-UCM,  C/Nicol\'as
Cabrera 13-15, 28049 Madrid, Spain} \email{mdeleon@icmat.es}

\author[\tiny Marrero]{Juan C.\ Marrero}
\address{Juan C.\ Marrero:
Unidad asociada ULL-CSIC, Geometr\'ia diferencial y mec\'anica geom\'etrica, Departamento de Matem\'atica Fundamental,
Facultad de Ma\-te\-m\'a\-ti\-cas, Universidad de
la Laguna, La Laguna, Tenerife, Canary Islands,
Spain} \email{jcmarrer@ull.es}

\author[\tiny Mart\'{\i}n de Diego]{David Mart\'{\i}n de Diego}
\address{D. Mart{\'\i}n de Diego:
Instituto de Ciencias Matem\'aticas, CSIC-UAM-UC3M-UCM,  C/Nicol\'as
Cabrera 13-15, 28049 Madrid, Spain} \email{david.martin@icmat.es}

\author[\tiny Mu\~noz]{Miguel C. Mu\~noz-Lecanda}
\address{
 Miguel C. Mu\~noz-Lecanda:
 Departamento de Matem\'atica Aplicada IV,   Edificio C-3, Campus Norte UPC,
   C/ Jordi Girona 1. 08034 Barcelona, Spain}
\email{matmcml@ma4.upc.edu}

\thanks{This work has been partially supported by MICINN (Spain)
Grants MTM 2007-62478, MTM2008-00689, MTM2009-08166 and MTM2009-13383; SOLSUVC 2008-01000238 and ProID20100210 of the Canarian  government, 2009SGR1338 of
the Catalan government, IRSES project GEOMECH (246981) within the 7th European Community Framework Program,
and project ``Ingenio
Mathematica" (i-MATH) No. CSD 2006-00032 (Consolider-Ingenio 2010). MBL has been financially supported by Juan de la Cierva fellowship from MICINN}
%
%\keywords{Discrete mechanics, Classical Field theories, Lie
%algebroids, Lagrangian Mechanics, Hamiltonian Mechanics.}
%
%
%
%
%

\begin{abstract} A close relationship between the classical Hamilton-Jacobi theory and
the kinematic reduction of control systems by decoupling vector
fields is shown in this paper. The geometric interpretation of this
relationship relies on new mathematical techniques for mechanics
defined on a skew-symmetric algebroid. This geometric structure
allows us to describe in a simplified way the mechanics of
nonholonomic systems with both control and external forces.
 %We will see that we can use the standard techniques of riemannian geometry  for free lagrangian systems to more involved situations.

\vspace{3mm}

\begin{center}
\textit{Dedicated to Tudor Ratiu on the occasion of his 60th birthday}
\end{center}
\end{abstract}
\date{\today}
\maketitle

\section{Introduction}

The reduction of mechanical control systems to kinematic
systems is very interesting and useful for solving control
problems such as optimal control problems
\cite{MariaMiguel} and for designing suitable control laws
(see \cite{MuYa}, Chapter 8 in \cite{bullolewis} and
references therein). For instance, the planning motion for
the associated kinematic system determines trajectories of
the mechanical control system. Thus the methodologies to
find these trajectories have been simplified because the
kinematic reduction gives rise to a first-order
control-linear system defined on the configuration
manifold. Hence in general it is easier to solve or to
analyze the kinematic system. If the mechanical control
system is reducible to a kinematic one, then the controlled
trajectories of this kinematic system under
reparametrization define solutions of the original
second-order problem on the phase space. An interesting
particular case is the one defined by kinematic reductions
of order $1$. This kind of reductions define a decoupling
vector field. Unfortunately, there is not a systematic
procedure for finding such kinematic reductions.

The philosophy of kinematic reductions of order 1 seems, in a first
approach, quite similar to the standard Hamilton-Jacobi theory. This
theory, that  appeared with the dawn of analytical mechanics, is a
valuable tool for the exact integration  of Hamilton's equations,
for instance using the technique of separation of variables (see
\cite{AbMa} and references therein). In many cases,  the
Hamilton-Jacobi theory allows us to simplify the integration of
Hamilton's equations or, at least, to find some particular
solutions. To be more precise, consider a configuration manifold $Q$
and a hamiltonian function $H: T^*Q\to \R$.  The Hamilton-Jacobi
equation can be written as
\[
H\left(q, \frac{\partial W}{\partial q}\right)=\hbox{constant}
\]
for some function $W: Q\to \R$. If we find such a function
$W$, then the integration of  the associated Hamilton's
equations (with initial conditions along ${\rm d} W(Q)$) is
reduced to knowing the integral curves of a vector field
$X_H^{{\rm d}W}$ on $Q$. This vector field is given by
$X_H^{{\rm d}W}=T\tau_{T^*Q}\circ X_H\circ {\rm d}W\in
{\mathfrak X}(Q)$, where $\tau_{T^*Q}: T^*Q\to Q$ is the
canonical projection and $X_H$ is the hamiltonian vector
field associated to $H$. Hence, from the integration of  a
vector field on the configuration space it is possible to
recover some of the solutions of the original hamiltonian
system. Recent developments in Hamilton-Jacobi theory are
described in \cite{CGMMMR,CaGrMaMaMuRo,ILM,LeMaMa,MaSo,BT}. Of
course, the possible similarities with the theory of
kinematic reductions are now clearer.

One of the main objectives in our paper is to carefully
study the underlying geometry of the kinematic reduction
theory by showing the close relation with the classical
Hamilton-Jacobi theory. Moreover, advantage of recent
developments in Hamilton-Jacobi theory for nonholonomic
systems on skew-symmetric algebroids \cite{BMMP,LeMaMa}
even with external forces  will be really useful to obtain
a full novel theory of kinematic reduction for this type of
systems. It is important to highlight this is not an
arbitrary generalization since the mechanics on algebroids
\cite{CoLeMaMar,CoMa,GGU0,TD69,Ma} is particularly relevant
for the class of Lagrangian systems invariant under the
action of a Lie group of symmetries including as a
particular case nonholonomic dynamics (see \cite{CoLeMaMaMa} for a
survey on the subject; see also \cite{LMM,Ma,Po2}).

The main results of this paper can be summarized in the following points:
\begin{itemize}
\item A description of nonholonomic mechanics in terms of the Levi-Civita connection associated to a fibered riemannian metric defined on the vector
subbundle determined by the nonholonomic constraints.
    \item A deduction of the Hamilton-Jacobi equation for nonholonomic systems in terms of the induced Levi-Civita
    connection.
    \item An affine connection approach in presence of control forces.
\item A description of Hamilton-Jacobi equation with controls.
    \item Relationship between Hamilton-Jacobi equation and kinematic reductions by decoupling vector
    fields.
    \end{itemize}

    It is interesting to observe that our approach allows us to extend
    the theory of kinematic reduction to controlled system with symmetries as for instance, nonholonomic Lagrange-Poincar\'e equations, etc.

In the sequel, all the manifolds are real, second countable and $C^\infty$. The maps are
assumed to be also $C^\infty$. Sum over all repeated indices is understood.

\section{Skew-symmetric algebroids}

In this section we introduce the notion of a skew-symmetric algebroid
on a vector bundle $\tau_D: D\to M$. It is known that this geometric
structure covers many interesting cases in mechanics, as for
instance, nonholonomic mechanics (see \cite{LeMaMa}). Similarly to
the intrinsic definition of the Euler-Lagrange equations for a
Lagrangian function $L: TM\to \R$ obtained by the canonical
structures on it (standard Lie bracket, exterior differential...),
it is possible to determine the motion equations for a Lagrangian
$L: D\to \R$ using the differential geometric structures naturally
induced by the skew-symmetric algebroid structure. We will show that
this generalization is quite useful in applications and clarifies
the dynamics of systems with nonholonomic constraints. Let us first
introduce the notion of a skew-symmetric algebroid.
\begin{definition}\label{alas}
A \emph{skew-symmetric algebroid structure} on the vector bundle
$\tau_{D}: D \to M$ is a $\R$-bilinear bracket $\lcf \cdot, \cdot
\rcf_{D}: \Gamma(\tau_{D}) \times \Gamma(\tau_{D}) \to
\Gamma(\tau_{D})$ on the space $\Gamma(\tau_{D})$ of sections of
$\tau_D$ and a vector bundle morphism $\rho_{D}: D \to TM$,
so-called \emph{anchor map}, such that:
\begin{enumerate}
\item
$\lcf \cdot, \cdot \rcf_{D}$ is skew-symmetric, that is,
\[
\lcf X, Y \rcf_{D} = -\lcf Y, X \rcf_{D}, \; \; \mbox{ for } X, Y
\in \Gamma(\tau_{D}).
\]

\item
If we also denote by $\rho_{D}: \Gamma(\tau_{D}) \to \mathfrak{X}(M)$
the morphism of $C^{\infty}(M)$-modules induced by the anchor map
then
\[
\lcf X, fY \rcf_{D} = f \lcf X, Y \rcf_{D} + \rho_{D}(X)(f) Y, \;
\; \mbox{ for } X, Y \in \Gamma(\tau_D) \mbox{ and } f \in
C^{\infty}(M).
\]

\end{enumerate}

If the bracket $\lcf \cdot, \cdot \rcf_{D}$ satisfies the Jacobi
identity, we have that the pair $(\lcf \cdot, \cdot \rcf_{D},
\rho_{D})$ is a \emph{Lie algebroid structure} on the vector
bundle $\tau_{D}: D \to M$.

\end{definition}

If $( \lcf \cdot, \cdot \rcf_{D}, \rho_{D})$ is a skew-symmetric
algebroid structure on the vector bundle $\tau_{D}: D \to M$, then
an \emph{almost differential} ${\rm d}^{D}$ of sections of
$\Lambda^k \tau_{D^*}$ being  $\tau_{D^*}\colon D^*\to M$  the vector bundle projection of the dual bundle $D^*$ is defined
as follows
\begin{equation*}\label{deD}
\begin{array}{rcl}
&&({\rm d}^{D}\alpha)(X_{0}, X_{1}, \dots, X_{k}) = \displaystyle
\sum_{i=0}^{k} (-1)^{i} \rho_{D}(X_{i})(\alpha(X_{0}, \dots,
\hat{X}_{i}, \dots, X_{k})) \\[5pt]
&&\qquad  + \displaystyle \sum_{i < j} (-1)^{i+j} \alpha (\lcf X_{i},
X_{j} \rcf_{D}, X_{0}, X_{1}, \dots, \hat{X}_{i}, \dots,
\hat{X}_{j}, \dots, X_{k})
\end{array}
\end{equation*}
for $\alpha \in \Gamma(\Lambda^{k}\tau_{D^*})$ and $X_{0}, X_1, \dots,
X_{k} \in \Gamma(\tau_{D})$.

In general $({\rm d}^D)^2\not=0$. Indeed,  $( \lcf \cdot,
\cdot \rcf_{D}, \rho_{D})$ is a Lie algebroid structure on
the vector bundle $\tau_{D}: D \to M$ if and only if $({\rm
d}^D)^2=0$ (see \cite{TD69,Mac,Ma,Weinstein99} for more
details about the Lie algebroids).

Suppose that $(x^i)$ are local coordinates on $M$ and that
$\{e_{A}\}$ is a local basis of the space of sections $\Gamma(\tau_D)$,
then
\begin{equation*} \lcf e_{A}, e_B\rcf_D =
{\mathcal C}^C_{AB} e_{C}, \ \ \ \rho_D (e_{A})=(\rho_D)_{A}^i \frac{\partial}{\partial x^i}  .\label{coeff-estruct}
\end{equation*}
The functions ${\mathcal C}^C_{AB}, (\rho_D)_A^i\in C^{\infty}(M)$
are called the \emph{local structure functions} of the
skew-symmetric algebroid on $\tau_D: D \rightarrow M$.

If $\{e^{A} \}$  is the dual basis of
$\{e_A\}$, then
\begin{eqnarray*}
{\rm d}^{D}F &=& (\rho_{D})_{A}^i \frac{\partial F}{\partial x^i} e^{A}, \\
{\rm d}^D\kappa &=& \left\{ (\rho_{D})_{A}^i \frac{\partial
\kappa_{B}}{\partial x^i}
   -\frac{1}{2}  {\mathcal C}_{AB}^{C} \kappa_{C}  \right\} e^{A} \wedge e^{B},
  \end{eqnarray*}
where $F\in C^{\infty}(M)$ and $\kappa=\kappa_{B} e^B \in
\Gamma(\tau_{{D}^*})$.

%\subsection{Differential equation associated to a section}

A $\rho_D$\emph{-admissible curve} is a curve $\gamma: I\subseteq
\mathbb{R} \longrightarrow D$  such that
\[
\frac{{\rm d}(\tau_D\circ \gamma)}{{\rm d}t}(t)=\rho_D(\gamma(t))\;
.
\]
Given $X\in \Gamma(\tau_D)$, the \emph{integral curves}
of the section $X$ are those curves $\sigma: I\subseteq
\R\rightarrow M$ such that satisfy
\[
\dot{\sigma}=\rho_D(X)\circ \sigma.
\]
That is, they are the integral curves of the associated vector field
$\rho_D(X)\in {\mathfrak X}(M)$. If $\sigma$ is an integral curve of
$X$, then $X\circ \sigma$ is a $\rho_D$-admissible curve. Locally,
the integral curves are characterized as the solutions of the
following system of differential equations
\[
 \dot{x}^i=(\rho_D)^i_A X^A(x),
\]
where $X=X^Ae_A$.

Consider now  the vector space over $\R$
\[
H^0(d^D) = \{f \in C^{\infty}(M) \; | \; {\rm d}^D f = 0\}.
\]
If  $M$ is connected and $D$ is
a transitive skew-symmetric algebroid, that is,
\[
\rho_D(D_{x}) = T_{x}M \; \; \mbox{ for all } x\in M,
\]
then
\begin{equation}\label{Poincare}
H^0({\rm d}^D) \simeq \R.
\end{equation}
It is important to stress that condition  (\ref{Poincare}) holds if
the skew-symmetric algebroid has a connected base space and is
completely nonholonomic, that is,
\begin{equation*}
{\rm Lie}^{(\infty)}_x(\rho_D(D))=T_xM
\end{equation*}
for all $x\in M$. See \cite{LeMaMa} for more details.

\section{Bundle metrics on skew-symmetric algebroids and Newtonian systems}\label{section3}

\subsection{The Levi-Civita connection}

Let ${\mathcal G}^D: D\times_M D\to \R$ be a nondegenerate bundle metric on a
skew-symmetric algebroid $(D,  \lcf \cdot, \cdot \rcf_{D},
\rho_{D})$. Given this bundle metric we can construct a unique
torsion-less connection $\nabla^{{\mathcal G}^D}$ on $D$ which is
metric with respect to ${\mathcal G}$ (see \cite{CoMa} and
references therein, for the standard case of Lie algebroids). The
following construction mimics the classical construction of the
Levi-Civita connection for a riemannian metric on a differentiable
manifold.

We will denote by $\flat_{{\mathcal G}^D}: D \to D^*$ the vector
bundle isomorphism induced by ${{\mathcal G}^D}$ and by $\#_{{\mathcal G}^D}: D^* \to D$ the inverse morphism.
The bundle metric can
be locally written as ${{\mathcal G}^D} = ({{\mathcal G}^D})_{AB} e^A\otimes e^B$.

The \emph{Levi-Civita connection} $\nabla^{{\mathcal G}^D}:
\Gamma(\tau_D)\times \Gamma(\tau_D)\to \Gamma(\tau_D)$ associated
to the bundle metric ${{\mathcal G}^D}$ is defined  by the
formula:
\[
\begin{array}{rcl}
2 {{\mathcal G}^D}(\nabla_{X}^{{\mathcal G}^D} Y, Z) & = &
\rho_{D}(X)({{\mathcal G}^D}(Y, Z)) + \rho_{D}(Y)({{\mathcal G}^D}(X, Z))\\
&&- \rho_{D}(Z)({{\mathcal G}^D}(X, Y))  +  {{\mathcal G}^D}(X, \lcf Z, Y\rcf_{D})\\
&& + {{\mathcal G}^D}(Y,\lcf Z,
X\rcf_{D}) - {{\mathcal G}^D}(Z, \lcf Y, X\rcf_{D})
\end{array}
\]
for $X, Y, Z \in \Gamma(\tau_D)$.

Alternatively, $\nabla^{{\mathcal G}^D}$ is determined by
\begin{equation*}\label{levi}
\begin{array}{l}
 \lcf X,Y\rcf_D= \nabla^{{\mathcal G}^D}_X Y-\nabla^{{\mathcal G}^D}_Y X \hbox{   (symmetry)}\\
 \rho_D(X)({{\mathcal G}^D}(Y,Z))={{\mathcal G}^D}(\nabla^{{\mathcal G}^D}_X Y, Z)+{{\mathcal G}^D}(Y, \nabla^{{\mathcal G}^D}_X Z)\hbox{   (metricity)}\; ,
 \end{array}
 \end{equation*}
These two properties allow to determine \emph{Christoffel symbols}
of the connection $\nabla^{{\mathcal G}^D}$ that satisfy
\[
\nabla^{{\mathcal G}^D}_{e_B}{e_C}=\Gamma^A_{BC}e_A.
\]
More details about how to compute the Christoffel symbols are given
in Section~\ref{Sec:example} if a ${\mathcal G}^D$-orthogonal basis
of $\Gamma(\tau_D)$ is taken.

Additionally, we have the notion of derivative along an admissible curve. If $\gamma\colon I \subseteq \mathbb{R} \rightarrow D$ is a $\rho_D$-admissible curve and \begin{equation*}\Gamma(\gamma)=\{X\colon I \subseteq \mathbb{R} \rightarrow D \, | \, X \mbox{ is } {\mathcal C}^\infty \mbox{ and } X(t)\in D_{\gamma(t)} \; \forall \; t\in I\}\end{equation*}
is the set of sections along $\gamma$, then the induced \emph{covariant derivative} $\nabla_\gamma\colon \Gamma(\gamma) \rightarrow \Gamma(\gamma)$ can be defined as the mapping from $X\in \Gamma(\gamma)$ to $\nabla_\gamma X \in \Gamma(\gamma)$ with local expression
\begin{equation}\label{3.0} \nabla_\gamma X= \left[ \ds{\frac{{\rm d} X^C}{{\rm d}t}+\Gamma^C_{AB}y^AX^B}\right]e_C,
\end{equation}
if $\gamma(t)=(x^i(t),y^A(t))$ and $X=X^Ae_A$.
% From the definition of  the Levi-Civita
% connection and assuming that the basis $\{e_A\}$ is ${\mathcal
% G}^D$-orthonormal, it is easy to deduce that
% \begin{equation}
% \Gamma^A_{BC}=\frac{1}{2}\left({\mathcal C}^C_{BA}+{\mathcal
% C}^B_{AC}+{\mathcal C}^A_{BC}\right)\; .
% \label{eq:ChristoffelSymbols}
% \end{equation}

\subsection{Geodesics}
Given the bundle metric ${\mathcal G}^D$, a $\rho_D$-admissible curve $\gamma: I \subseteq \mathbb{R} \rightarrow D$ on $D$ is said to be a \emph{geodesic} if
\[
\nabla_{\gamma}\gamma = 0.
\]
If the local expression of $\gamma$ is
\[
\gamma(t) = (x^i(t), y^{A}(t)),
\]
then $\gamma$ is a geodesic if and only if
\[
\displaystyle \frac{dx^{i}}{dt} = (\rho_D)^{i}_{C}y^{C}, \; \; \frac{dy^{C}}{dt} = -\Gamma_{AB}^{C} y^{A}y^{B}.
\]
The geodesics are just the integral curves of a vector field on $D$, called \emph{the geodesic spray} $\xi_{{\mathcal G}^{D}}$, whose local expresion is
\[
\xi_{{\mathcal G}_{D}} = \displaystyle (\rho_{D})_{C}^{i} y^{C} \frac{\partial}{\partial x^{i}} - \Gamma_{AB}^{C} y^{A} y^{B} \frac{\partial}{\partial y^{C}}.
\]
Note that if $a \in D$, then there exists a unique geodesic $\gamma_a: (-\epsilon, \epsilon) \subseteq \mathbb{R} \rightarrow D$ such that $\gamma_a(0) = a$. If $\sigma_a = \tau_D \circ \gamma_a: (-\epsilon, \epsilon) \subseteq \mathbb{R} \rightarrow M$ is the base curve of $\gamma_a$ then, since $\gamma_a$ is $\rho_D$-admissible, we have that
\[
\dot{\sigma}_a(t) = \rho_D(\gamma_a(t)), \; \; \forall t \in I.
\] 
The associated
\emph{symmetric product} is defined as follows:
\[
\langle X: Y\rangle_{{\mathcal G}^D}=\nabla^{{\mathcal G}^D}_X Y+ \nabla^{{\mathcal G}^D}_Y X\; ,\quad X, Y\in \Gamma(\tau_D)\; .
\]
The symmetric product on Riemannian manifolds is a fundamental tool in controllability
results, kinematic reduction of mechanical systems and in the
characterization of geodesic invariance of distributions (see
\cite{bullolewis}). These results can be extended to our setting of skew-symmetric algebroids.

\begin{lemma} If $X, Y \in \Gamma(\tau_D)$,  then $\langle X: Y\rangle_{{\mathcal G}^D}^{\bf v}=[X^{\bf v},[\xi_{{\mathcal G}_{D}}, Y^{\bf v}]]$, where $X^{\bf v}$ is the natural vertical lift of the section $X$ of $\tau_D$.
\begin{proof}
This is proved locally, similarly to the proof of this result on Riemannian manifolds in Lemma B.3 in~\cite{bullolewis}.

Locally, $X=X^Ae_A$ and $Y=Y^Be_B$, then
\begin{equation*} X^{\bf v}=\ds{X^A \, \frac{\partial}{\partial y^A}, \quad  Y^{\bf v}=Y^A \, \frac{\partial}{\partial y^A}}.
\end{equation*}
Let us compute,
\begin{eqnarray*}
[\xi_{{\mathcal G}_{D}}, Y^{\bf v}]&=& \ds{ \left[  (\rho_{D})_{C}^{i} y^{C} \frac{\partial}{\partial x^{i}} - \Gamma_{AB}^{C} y^{A} y^{B} \frac{\partial}{\partial y^{C} },Y^E \frac{\partial}{\partial y^{E} }\right]}\\
&=& \ds{  (\rho_{D})_{C}^{i} y^{C} \frac{\partial Y^E}{\partial x^{i}} \frac{\partial}{\partial y^{E} } -
Y^B (\rho_D)^i_B \frac{\partial }{\partial x^i}}\\&+& \ds{Y^B(\Gamma_{AB}^{C}+\Gamma_{BA}^{C}) y^{A} \frac{\partial}{\partial y^{C} }}= \ds{-Y^B (\rho_D)^i_B \frac{\partial }{\partial x^i}}\\&+&\ds{ \left(  (\rho_{D})_{A}^{i} y^{A}  \frac{\partial Y^C}{\partial x^{i}}+ Y^B(\Gamma_{AB}^{C}+\Gamma_{BA}^{C}) y^{A}  \right)   \frac{\partial}{\partial y^{C} }}.
\end{eqnarray*}
Then, 
\begin{eqnarray*}
[X^{\bf v},[\xi_{{\mathcal G}_{D}}, Y^{\bf v}]]&=& \ds{ \left[ X^A\frac{\partial }{\partial y^A}, -Y^B (\rho_D)^i_B \frac{\partial }{\partial x^i}+ \Bigg(  (\rho_{D})_{A}^{i} y^{A}  \frac{\partial Y^C}{\partial x^{i}}\right.}\\&+&\ds{\left. Y^B(\Gamma_{AB}^{C}+\Gamma_{BA}^{C}) y^{A}  \Bigg)   \frac{\partial}{\partial y^{C} }\right]= \left(  X^A(\rho_{D})_{A}^{i}  \frac{\partial Y^C}{\partial x^{i}}\right. }\\
&+& \ds{\left.Y^B(\rho_D)^i_B \frac{\partial X^C}{\partial x^i}+ X^AY^B(\Gamma_{AB}^{C}+\Gamma_{BA}^{C})  \right)   \frac{\partial}{\partial y^{C} }}.
\end{eqnarray*}
On the other hand,
\begin{eqnarray*}
\langle X: Y\rangle_{{\mathcal G}^D}&=&\ds{\nabla^{{\mathcal G}^D}_X Y+ \nabla^{{\mathcal G}^D}_Y X=X^A (\rho_D)^i_A \frac{\partial Y^C}{\partial x^i} e_C+ X^A Y^B \Gamma^C_{AB} e_C } \\
&+& \ds{Y^B (\rho_D)^i_B \frac{\partial X^C}{\partial x^i} e_C+ X^A Y^B \Gamma^C_{BA} e_C =\left(X^A (\rho_D)^i_A \frac{\partial Y^C}{\partial x^i}\right.}  \\
&+& \ds{\left.Y^B (\rho_D)^i_B \frac{\partial X^C}{\partial x^i}
+ X^A Y^B(\Gamma^C_{AB}+ \Gamma^C_{BA} ) \right) e_C }.
\end{eqnarray*}
This proves the equality since the vertical lift of this section is equal to the local expression of $[X^{\bf v},[\xi_{{\mathcal G}_{D}}, Y^{\bf v}]]$ we just computed above.
\end{proof} \label{Lemma:SymProdV}
\end{lemma}

Now, we can extend the characterization of geodesically invariant distributions already known on Riemannian manifolds~\cite{bullolewis} to skew-symmetric algebroids. A subbundle $\mathfrak{D}$ of a skew-symmetric algebroid $(D, \lcf \cdot, \cdot \rcf_D, \rho_D)$  is \textit{geodesically invariant} if for any geodesic $\gamma\colon I \subseteq  \mathbb{R} \rightarrow D$ with initial condition $\gamma(0)\in \mathfrak{D}(\sigma(0))$, then $\gamma(t)\in \mathfrak{D}(\sigma(t))$ for any $t\in I$, where $\sigma=\tau_D \circ \gamma$.

\begin{theorem} Let $(D, \lcf \cdot, \cdot \rcf_D, \rho_D)$ be a skew-symmetric algebroid and ${\mathcal G}_{D}$ be a bundle metric over $D$. Let $\mathfrak{D}$ be a subbundle of $D$ and $\tau_{\mathfrak{D}}={\tau_D}_{|\mathfrak{D}}$.  The following statements are equivalent:
\begin{enumerate}
\item[(i)]  $\mathfrak{D}$  is geodesically invariant.
\item[(ii)] The restriction of the geodesic spray $\xi_{{\mathcal G}^D}$ to $\mathfrak{D}$ is tangent to $\mathfrak{D}$.
\item[(iii)] If $X,Y \in \Gamma(\tau_{\mathfrak{D}})$, then $\langle X, Y \rangle_{{\mathcal G}^D}\in \Gamma(\tau_{\mathfrak{D}})$.
\end{enumerate}
\begin{proof}
$(i)\Leftrightarrow (ii)$ The integral curves of  $\xi_{{\mathcal G}^D}$ are the geodesics. This proves the equivalence.

$(ii)\Rightarrow (iii)$ Assume that  $X,Y \in \Gamma(\tau_{\mathfrak{D}})$.  Then the restrictions of $X^{\bf v}$, $Y^{\bf v}$ to $\mathfrak{D}$ are tangent to $\mathfrak{D}$.  By hypothesis, $(\xi_{{\mathcal G}^D})_{\left. \right|_{\mathfrak{D}}}$ is tangent to $\mathfrak{D}$. Thus, $[\xi_{{\mathcal G}^D},Y^{\bf v}]_{{\left. \right|}_{\mathfrak{D}}}$ is tangent to $\mathfrak{D}$. As $X^{\bf v}_{\left. \right|_{\mathfrak{D}}}$ is also tangent to $\mathfrak{D}$, we have that $[X^{\bf v},[\xi_{{\mathcal G}^D},Y^{\bf v}]]_{\left. \right|_{\mathfrak{D}}}$ is tangent to $\mathfrak{D}$.  By Lemma~\ref{Lemma:SymProdV}, $([X^{\bf v},[\xi_{{\mathcal G}^D},Y^{\bf v}]])_{\left. \right|_{\mathfrak{D}}}=(\langle X \colon Y \rangle^{\bf v}_{{{\mathcal G}^D}})_{\left. \right|_{\mathfrak{D}}}$. Thus $\langle X \colon Y \rangle_{{\mathcal G}^D}\in \Gamma(\tau_{\mathfrak{D}})$ because  $X^{\bf v}$ restricted to  $\mathfrak{D}$ are tangent to $\mathfrak{D}$ if and only if $X\in \Gamma(\tau_{\mathfrak{D}})$.

$(iii)\Rightarrow (i)$ Let $\{X_1,\dots, X_r \}=\{X_a\}_{a=1,\dots,r}$ be a local basis for $\Gamma(\tau_{\mathfrak{D}})$. It can be extended into a local basis for $\Gamma(\tau_D)$, which is given by \begin{equation*} \{X_1,\dots, X_r,X_{r+1},\dots, X_n\}=\{X_a,X_\alpha\},\end{equation*} for  $a=1,\dots, r$; $\alpha=r+1,\dots , n$. Hence, 
\begin{equation*}
\nabla^{{\mathcal G}^D}_{X_a}X_b = \Gamma^c_{ab}X_c+\Gamma^\alpha_{ab}X_\alpha, \quad {\rm for } \quad a,b\in \{1,\dots,r\}.
\end{equation*}
By assumption, $\langle X_a \colon X_b \rangle_{{\mathcal G}^D}\in \Gamma(\tau_{\mathfrak{D}})$, then $\Gamma^\alpha_{ab}+\Gamma^\alpha_{ba}=0$ for all $a,b\in \{1,\dots, r\}$, $\alpha\in \{r+1,\dots, n\}$ because
\begin{equation*}
\langle X_a \colon X_b \rangle_{{\mathcal G}^D}= \nabla^{{\mathcal G}^D}_{X_a}X_b+\nabla^{{\mathcal G}^D}_{X_b} X_a=(\Gamma^c_{ab}+\Gamma^c_{ba})X_c+(\Gamma^\alpha_{ab}+\Gamma^\alpha_{ba})X_\alpha.
\end{equation*}

Now, let $\gamma\colon I\rightarrow D$ be a geodesic such that $\gamma(0)\in \mathfrak{D}(\sigma(0))$ where $\sigma=\tau_D \circ \gamma$.  Suppose that 
\begin{equation*}
\gamma(t)=\sum_{a=1}^ru^a(t)X_a(\sigma(t))+\sum_{\alpha=r+1}^n u^\alpha(t) X_\alpha(\sigma(t))=\sum_{i=1}^n u^i(t)X_i(\sigma(t)),
\end{equation*}
then $u^\alpha(0)=0$ for all $\alpha\in \{r+1,\dots,n\}$. As $\gamma$ is a geodesic,
\begin{equation*} 0=\nabla^{{\mathcal G}^D}_{\gamma(t)} \gamma(t)=\ds{\sum_{k=1}^n \left(\frac{{\rm d} u^k} {{\rm d}t} +\Gamma^k_{ij}(\sigma(t))u^i(t)u^j(t)\right)X_k(\sigma(t))}.
\end{equation*}
Then
\begin{equation}
\ds{\frac{{\rm d} u^k} {{\rm d}t} +(\Gamma^k_{ij} \circ \sigma)u^iu^j=0}, \quad \forall \; k\in \{1,\dots,n\}. \label{eq:geodesic}
\end{equation}
As  $\Gamma^\alpha_{ab}+\Gamma^\alpha_{ba}=0$ for all $a,b\in \{1,\dots, r\}$, $\alpha\in \{r+1,\dots, n\}$, we have that $u^\alpha(t)=0$ for all $t\in I$ is a solution of the differential equation~\eqref{eq:geodesic} for $\alpha=r+1,\dots, n$:
\begin{eqnarray*}
\ds{\frac{{\rm d} u^\alpha} {{\rm d}t} +(\Gamma^\alpha_{ij} \circ \sigma)u^iu^j}&=&\ds{\frac{{\rm d} u^\alpha} {{\rm d}t} +(\Gamma^\alpha_{ab} \circ \sigma)u^au^b +(\Gamma^\alpha_{a\beta} \circ \sigma)u^au^\beta}\\&+&\ds{(\Gamma^\alpha_{\delta b} \circ \sigma)u^\delta u^b+(\Gamma^\alpha_{\delta \beta} \circ \sigma)u^\delta u^\beta}\\
&=& (\Gamma^\alpha_{ab} \circ \sigma)u^au^b=0.
\end{eqnarray*}
As the initial condition for $u^\alpha$ is $u^\alpha(0)=0$ for all $\alpha\in \{r+1,\dots, n\}$, it follows that $u^\alpha(t)=0$ is the unique solution for the differential equations in~\eqref{eq:geodesic}. Thus,
\begin{equation*} \gamma(t)=\sum_{a=1}^r u^a(t)X_a(\sigma(t)) \in \mathfrak{D}(\sigma(t)), \quad \forall \; t\in I\subseteq \mathbb{R}.
\end{equation*}
Hence $\mathfrak{D}$ is geodesically invariant.
\end{proof}
\label{Thm:GeodInv}
\end{theorem}
\subsection{Newtonian systems}
Given a bundle map ${\mathcal F}: D\rightarrow D$ (that is,
$\tau_D\circ {\mathcal F}=\tau_D$) we define a \emph{newtonian
system} as the triple $(D, {\mathcal G}^D, {\mathcal  F})$. This
newtonian system induces the system of differential equations:
\begin{equation}\label{noholonoma-forced}
\nabla^{{\mathcal G}^D}_{\gamma(t)}\gamma(t)=  {\mathcal F}(\gamma(t)), \qquad t\in I\; ,
\end{equation}
where the solutions are curves $\gamma: I\subseteq \R\rightarrow D$
which are  $\rho_D$-admissible.

 Given local coordinates $(x^i, y^A)$ associated with the basis $\{e_A\}$ for sections of $D$,  Equations
   (\ref{noholonoma-forced}) can be written as
\begin{equation}
\begin{array}{rcl}
\dot{x}^i&=&(\rho_D)^i_A y^A,\\
\dot{y}^C&=&-\Gamma^C_{AB}y^Ay^B+ {\mathcal F}^C(x,y),
\end{array} \label{eqo}
\end{equation}
where  ${\mathcal F}(x^j, y^B)=(x^j, {\mathcal F}^A (x^j, y^B))$.

Observe that Equations (\ref{eqo}) are the equations of the integral
curves of a vector field $\xi_{{\mathcal G}^D, {\mathcal F}}$ on
$D$. Locally, this vector field is given by
\[
\xi_{{\mathcal G}^D, {\mathcal F}}=(\rho_D)^i_A y^A\frac{\partial}{\partial x^i}+\left(-\Gamma^C_{AB}y^Ay^B+ {\mathcal F}^C\right)\frac{\partial}{\partial y^C}\; .
\]

\begin{remark} \label{Remark:FgradV}
 The map ${\mathcal F}$ could be given by a section $F\in
\Gamma(\tau_D)$ such that ${\mathcal F}=F\circ \tau_D$. An
interesting particular case is  when we have a potential function
$V: M\to \R$ and $F$ is the section $-\hbox{grad}_{{\mathcal
G}^D}V\in \Gamma(\tau_{D})$ given by
\[
{{\mathcal G}^D}(\hbox{grad}_{{\mathcal G}^D}V, X) =
\rho_{D}(X)(V)={\rm d}V(\rho_D(X)), \; \; \mbox{ for  every } X \in
\Gamma(\tau_{D}).
\]
In particular, the solutions of this newtonian system $(D, {\mathcal
G}^D, F=\linebreak -\hbox{grad}_{{\mathcal G}^D}V)$ are equivalent to the
solutions of the Euler-Lagrange equations on a skew-symmetric algebroid
with Lagrangian $L: D\longrightarrow \R$:
\begin{equation}\label{lagrangian}
L(v)=\frac{1}{2}{{\mathcal G}^D}(v,v)-V(\tau_D(v)).
\end{equation}
Therefore, these solutions are
 $\rho_D$-admissible curves
$\gamma: I\longrightarrow D$  such that
\begin{equation}\label{noholonoma}
\nabla^{{\mathcal G}^D}_{\gamma(t)}\gamma(t) + \hbox{grad}_{{\mathcal
G}^D}V(\tau_D(\gamma(t)))=0.
\end{equation}
Locally, those solutions satisfy
\begin{equation*}
\begin{array}{rcl}
\dot{x}^i&=&(\rho_D)^i_A y^A,\\
\dot{y}^C&=&\ds{-\Gamma^C_{AB}y^Ay^B-({{\mathcal
G}^D})^{CB}(\rho_D)^i_B\frac{\partial V}{\partial x^i}}\; ,
\end{array} \label{equa1}
\end{equation*}
where $({{\mathcal G}^D})^{AB}$ are the entries of the inverse matrix of $(({{\mathcal G}^D})_{AB})$.
In this case, the vector field $\xi_{{\mathcal G}^{D}, F} $ on $D$ is given by
$\xi_{{\mathcal G}^{D}, V} = \xi_{{\mathcal G}^{D}} + (grad_{{\mathcal G}^{D}}V)^{\bf v}$, where  $(grad_{{\mathcal G}^{D}}V)^{\bf v}$
is the vertical lift to $D$ of the section $ grad_{{\mathcal G}^{D}}V \in \Gamma(\tau_D)$.
\end{remark}

\begin{example}\label{example:DequalTM}
{\rm If $D=TM$, $\lcf\, ,\rcf_D=[\, ,\, ]$ the standard Lie bracket
on $M$, $\rho_D=\hbox{Id}_{TM}$ and ${{\mathcal G}^D}$ is a
riemannian metric on $M$, then Equations (\ref{noholonoma}) are the
classical \emph{Euler-Lagrange equations} for the mechanical
lagrangian $L: TM\longrightarrow \R$.
}
\end{example}

\begin{example}\label{noholo-example}
{\rm
Given a regular distribution $D$ on $TM$  and a riemannian metric ${\mathcal G}^{TM}$, we  consider
 the riemannian orthogonal decomposition $TM = D \oplus
D^{\perp}$ and the associated orthogonal projectors $P: TM \to D$
and $Q: TM\to D^{\perp}$, see \cite{BMMP,LeMaMa}. Denote also by
$\iota_D: D\hookrightarrow TM$ the canonical inclusion. We induce by
restriction a bundle metric ${\mathcal G}^D: D\times_M D\to \R$ and
an skew-symmetric algebroid structure on $D$ as follows:
\[
\lcf X, Y\rcf_D=P[\iota_D(X), \iota_D(Y)]\; ,\ \ \
\rho_D(X)=\iota_D(X)\, ,
\]
where $X, Y\in \Gamma (\tau_D)$. Note that in this example, $X,Y$ are vector
fields on $M$ taking values on $D$. Moreover, the Levi-Civita connection $\nabla^{{\mathcal G}^D}$
coincides with the constrained connection
$\nabla^D_XY\colon =P(\nabla^{{\mathcal G}}_XY)$
defined, for instance, in~\cite{bullolewis} if $\nabla^D$  is
restricted to $\Gamma(\tau_D)$.

For this particular skew-symmetric algebroid structure, Equations
(\ref{noholonoma}) correspond with the equations of the nonholonomic
system determined by the constraints induced by the distribution $D$
and the mechanical lagrangian (\ref{lagrangian}). These equations
are also called  in the literature \textit{Lagrange-D'Alembert's
equations}.

Consider a basis of  ${\mathcal G}^{D}$-orthogonal vector fields
$\{X_A, X_{\alpha}\}$, $1\leq A \leq m={\rm rank}\, D$, $m+1\leq
\alpha \leq n=\dim M$, adapted to the decomposition $TM = D \oplus
D^{\perp}$. In other words, $D_x=\hbox{span }\{X_A(x)\}$ and
$D^\perp_x=\hbox{span }\{X_{\alpha}(x)\}$. Observe that for the
induced coordinates $(x^i,y^A,y^\alpha)$ on $TM$ the nonholonomic
constraints are rewritten as $y^{\alpha}=0$, $m+1\leq \alpha \leq
n$. That is, the induced coordinates on $D$ are given by
$(x^i,y^A)$. Therefore, the skew-symmetric algebroid structure induced on the
vector subbundle $D\rightarrow M$ is locally described by:
\begin{eqnarray*}
\lcf X_A, X_B\rcf_D&=&P[X_A, X_B]=P({\mathcal C}_{AB}^C X_C+ {\mathcal C}_{AB}^{\beta} X_{\beta})={\mathcal C}_{AB}^C X_C,\\
\rho_D(X_A)&=&X_A\; ,
\end{eqnarray*}
where $X_A=(\rho_D)^i_A\frac{\partial}{\partial x^i}$, $1\leq A \leq
 m$.

The \emph{Lagrange-D'Alembert's equations} are
\begin{eqnarray}
\dot{x}^i&=&(\rho_D)^i_A y^A, \nonumber \\
\dot{y}^C&=&-\Gamma^C_{AB}y^Ay^B-({{\mathcal
G}^D})^{CB}(\rho_D)^i_B\frac{\partial V}{\partial
x^i}\label{equa1LdA}\; .
\end{eqnarray}%and
% \[
% \Gamma^A_{BC}=\frac{1}{2}\left({\mathcal C}^C_{BA}+{\mathcal
% C}^B_{AC}+{\mathcal C}^A_{BC}\right)\; .
% \]

}
\end{example}

 \begin{example}
 {\rm
 Our theory is not only restricted to  lagrangian systems defined on the tangent bundle $TM$  or nonholonomic systems determined by a regular distribution  on $TM$. The techniques described in this paper by means of skew-symmetric algebroids are general enough to cover the most important cases of reduction of mechanical systems subjected or not to nonholonomic constraints.

As a particular example, we include in our analysis the case of Lie algebras ${\mathfrak g}$ of finite dimension (it is clear that ${\mathfrak g}$ is a Lie algebroid over a single
point). Now, suppose that $(l, {\mathfrak d})$ is a nonholonomic Lagrangian system on ${\mathfrak g}$, where
$l: {\mathfrak g}\to \R$ is a Lagrangian function defined by $l(\xi)=\frac{1}{2}\langle \I \xi, \xi\rangle$, $\I: {\mathfrak g}\rightarrow {\mathfrak g}^*$ is a symmetric positive
definite inertia operator and ${\mathfrak d}$  is a vector subspace of ${\mathfrak g}$.
We have the orthogonal decomposition
\[
{\mathfrak g}={\mathfrak d}\oplus {\mathfrak d}^{\perp},
\]
where ${\mathfrak d}^{\perp}=\{\xi'\in {\mathfrak g}\, |\, \langle \I \xi', \xi\rangle=0 \; \forall \xi \in {\mathfrak d}\}$.
Take now an adapted basis to this decomposition $\{e_A, e_{\alpha}\}$ where ${\mathfrak d}=\hbox{span } \{e_A\}$ and  ${\mathfrak d}^{\perp}=\hbox{span } \{e_{\alpha}\}$. Then, the  \emph{Euler-Poincar\'e-Suslov equations} for $(l, {\mathfrak d})$ are
\[
\dot{y}^C=-\Gamma^C_{AB}y^Ay^B, \; 
\]
where $\{y^{A}, y^{\alpha}\}$ are the global coordinates on ${\mathfrak g}$ induced by the basis $\{e_A, e_{\alpha}\}$.
}
 \end{example}

 \begin{example}
 {\rm
Similarly, more involved situations can be recovered using our techniques.
For instance, nonholonomic systems on Atiyah algebroids associated with principal $G$-bundles. For the sake of simplicity, we will consider the particular case when the principal $G$-bundle is trivial. In such a case, the Atiyah algebroid is a vector bundle of the form
\[
\tau_A: A={\mathfrak g}\times TM\longrightarrow M,
\]
where ${\mathfrak g}$ is the Lie algebra of the Lie group $G$ and $M$ is a smooth manifold. The Lie bracket of the space $\Gamma (\tau_A)$ is characterized by the following condition
\[
\lcf (\xi, X), (\xi', X')\rcf_A=(\lcf \xi, \xi'\rcf_{\mathfrak g}, [X, X'])\; ,
\]
for $\xi, \xi'\in {\mathfrak g}$ and $X, X'\in {\mathfrak X}(M)$. The anchor map $\rho_A$ is the canonical projection onto the second factor.

Suppose now that $D$ is a vector subbundle of $A$ over $M$ of constant rank (the constraint bundle) such that
\[
M \ni x\longrightarrow D_V(x)\colon =D(x)\cap ({\mathfrak g}\times \{0_{T_xM}\})\subseteq {\mathfrak g}\times T_xM
 \]
is a vector subbundle of $A$. Then we can choose a local basis $\{\xi_a\}_{1\leq a\leq r}$ of $\Gamma(\tau_{D_V})$, with $\xi_a: U\subseteq M\longrightarrow {\mathfrak g}$ smooth maps, and a local basis $\{X_A\}=\{\xi_a, (\eta_{\alpha}, Y_{\alpha})\}$ of $\Gamma(\tau_D)$, with $\eta_{\alpha}: U\subseteq M\longrightarrow {\mathfrak g}$ and $Y_{\alpha}\in {\mathfrak X}(U)$.

Moreover, if $(x^i)$ are local coordinates on $U\subseteq M$ and $Y_{\alpha}=Y_{\alpha}^i\frac{\partial}{\partial x^i}$, 
the \emph{Lagrange-D'Alembert-Poincar\'e equations} are: 
\begin{eqnarray*}
\dot{x}^i&=&Y_{\alpha}^i y^{\alpha},\\
\dot{y}^c&=&-\Gamma^c_{AB} y^A y^B-\frac{\partial V}{\partial x^i}Y^i_{\alpha} ({\mathcal G}^D)^{c\alpha},\\
\dot{y}^\alpha&=&-\Gamma^\alpha_{AB} y^A y^B-\frac{\partial V}{\partial x^i}Y^i_{\beta} ({\mathcal G}^D)^{\alpha\beta},
\end{eqnarray*}
where $(x^i, y^c, y^{\alpha})$ are the corresponding local coordinates on ${\mathcal D}$.

Note that in the particular case when $M$ is a single point, we recover the \emph{Euler-Poincar\'e Suslov equations}.

Geometric interpretations of \emph{nonholonomic LR} systems  or \emph{nonholonomic systems
with semidirect product symmetry} may also be given using skew-symmetric algebroids deduced from Lie algebroids (see \cite{CoLeMaMar} for more details).
}

\end{example}

\section{Hamilton-Jacobi equation}

The next result is a direct consequence of
Equations~\eqref{noholonoma}. See \cite{BMMP} for  an extension of
this result (in a hamiltonian context) for many different types of
mechanical systems (nonholonomic dynamics, dissipative systems...).

\begin{proposition}\label{propo1}
 Let $(D, \lcf\; ,\rcf_D, \rho_D)$ be a skew-symmetric algebroid and consider a newtonian system determined by
$(D,  {\mathcal G}^D, {\mathcal F})$.
Take an arbitrary  section $X\in \Gamma(\tau_D)$
then, the
following conditions on $X$  are equivalent:
\begin{enumerate}
\item
If $\sigma: I \longrightarrow M$ is an integral curve of the vector field $\rho_D(X)$
that is,
\begin{equation}\label{sigma-03}
\dot{\sigma}(t)=\rho_D(X)(\sigma(t)),
\end{equation}
then $\gamma=X\circ \sigma: I \longrightarrow D$ is  a solution of
\[
\nabla^{{\mathcal G}^D}_{\gamma(t)}\gamma(t)={\mathcal F}(\gamma(t)).
\]

\item
$X$ satisfies
\begin{equation*}\label{0campo}
\nabla^{{\mathcal G}^D}_X X ={\mathcal F} \circ X.
\end{equation*}
\end{enumerate}
\end{proposition}

\begin{proof}
Let $\sigma: I \rightarrow M$ be an integral curve of the vector field $\rho_D(X)$. Then, the result follows immediately from the fact that
\begin{equation*}
 \nabla^{{\mathcal G}^D}_{X\circ \sigma (t)}( X\circ \sigma(t)) =(\nabla^{{\mathcal G}^D}_X X ) \circ  \sigma(t).
\end{equation*}

% Observe that locally if $X=X^A e_A$   then
% \begin{eqnarray*}
% &&\left(\nabla^{{\mathcal G}^D}_X X +\hbox{grad}_{{\mathcal
% G}^D}V\right) (\sigma(t))\\
% &&=
% \left(\Gamma^C_{AB}(\sigma(t)) \gamma^A(t) \gamma^B(t)\right.\\
% &&\left.+\frac{d}{dt}\gamma^C(t)+({{\mathcal G}^D})^{CB}(\sigma(t))(\rho_D)^i_B(\sigma(t))\frac{\partial V}{\partial x^i}(\sigma(t)\right) e_C(\sigma(t))
% \end{eqnarray*}
% where $\sigma$ is the curve defined on  (\ref{sigma-03}),
% $\gamma(t)=X(\sigma(t))$ and, in consequence,
% $\gamma(t)=\gamma^A(t)e_A(\sigma(t))=X^A(\sigma(t)) e_A(\sigma(t))$.
% From (\ref{equa1LdA}) we directly deduce the equivalence of (i) and
% (ii).
\end{proof}

This result is analogous to the Hamilton-Jacobi theory already described
in a Lagrangian framework for a free mechanical system
in~\cite{CGMMMR} and for a nonholonomic one in~\cite{CaGrMaMaMuRo},
but adapted to skew-symmetric algebroid structures.

\begin{remark}
An analogous result can be written replacing the section $X$ by an
entire distribution spanned by sections.
\end{remark}

\begin{remark}
{\rm In standard riemannian geometry (that is, $D=TM$ equipped with
the standard Lie bracket and without external forces ${\mathcal
F}\equiv 0$) the vector fields $X$ satisfying $\nabla^{{\mathcal
G}^D}_X X=0$ are called auto-parallel vector fields and are
obviously connected with the solutions of the geodesic equations}.
\end{remark}

 From now on, we only consider mechanical problems given by $(D,  {\mathcal G}^D,
\linebreak V)$ as described in Remark~\ref{Remark:FgradV}. Specializing Proposition \ref{propo1} to this kind of problems and to
vector fields $X$ verifying an extra condition  $i_X{\rm d}^D(\flat_{\mathcal G}(X)) = 0$
 we obtain a new expression of this Proposition. Indeed, the following Theorem
can be compared with the classical expression of the Hamilton-Jacobi equation
 proposed in \cite{LeMaMa} for the hamiltonian function $h: D^*\longrightarrow \R$:
 \[
 h(\kappa)={\mathcal G}^{D^*}(\kappa, \kappa)+V(\tau_{D^*}(\kappa)),
 \]
 where $\kappa\in D^*$ and ${\mathcal G}^{D^*}: D^*\times_M D^*\longrightarrow \R$ is the induced bundle metric on the dual bundle.
 
\begin{theorem}\label{maintheorem}
Let $(D, \lcf\; ,\rcf_D, \rho_D)$ be a skew-symmetric algebroid and
consider a mechanical problem determined by $(D,  {\mathcal G}^D,
V)$. Take a section $X\in \Gamma(\tau_D)$ such that
$i_X{\rm d}^D(\flat_{\mathcal G^D}(X)) = 0$. Under this hypothesis, the
following conditions are equivalent:
\begin{enumerate}
\item
If $\sigma: I \longrightarrow M$ is an integral curve of the vector field $\rho_D(X)$,
that is,
\begin{equation*}\label{sigma}
\dot{\sigma}(t)=\rho_D(X)(\sigma(t)),
\end{equation*}
then $\gamma=X\circ \sigma: I \longrightarrow D$ is  a solution of
\[
\nabla^{{\mathcal G}^D}_{\gamma(t)}\gamma(t) + \hbox{grad}_{{\mathcal
G}^D}V(\tau_D(\gamma(t)))=0.
\]

\item
$X$ satisfies \emph{the Hamilton-Jacobi differential equation}
\begin{equation}\label{0h-j-p}
{\rm d}^{D}\left(\frac{1}{2}\mathcal G^D(X,X)+V\right) = 0.
\end{equation}
\end{enumerate}
If, additionally, the skew-symmetric algebroid $(D, \lcf \cdot ,
\rcf_{D}, \rho_{D})$ is completely nonholonomic and $M$ is connected
or if $H^0(d^D) \simeq \R$, then Equation (\ref{0h-j-p}) is
equivalent to
\[
\frac{1}{2}{\mathcal G}^D(X,X)+V = \hbox{constant}.
\]
\end{theorem}

\begin{proof}

First,  we study how to express the condition $i_X{\rm
d}^D(\flat_{{\mathcal G}^D}(X)) = 0$ in terms of the Levi-Civita
connection associated to ${{\mathcal G}^D}$. Let $Y\in
\Gamma(\tau_D)$,
\begin{eqnarray*}
0&=&{\rm d}^D(\flat_{{\mathcal G}^D}(X))(X,Y)\\
 &=& \rho_D(X)({{\mathcal G}^D}(X,Y))-\rho_D(Y)({{\mathcal G}^D}(X,X))-{{\mathcal G}^D}(X, \lcf X, Y\rcf_D)\\
 &=&{{\mathcal G}^D}(\nabla^{{\mathcal G}^D}_X X, Y)+{{\mathcal G}^D}(\nabla^{{\mathcal G}^D}_X Y, X)-{{\mathcal G}^D}(\nabla^{{\mathcal G}^D}_Y X, X)-{{\mathcal G}^D}(\nabla^{{\mathcal G}^D}_Y X, X)\\
 &&-{{\mathcal G}^D}(X, \nabla^{{\mathcal G}^D}_X Y-\nabla^{{\mathcal G}^D}_Y X)\\
 &=&{{\mathcal G}^D}(\nabla^{{\mathcal G}^D}_X X, Y)-{{\mathcal G}^D}(\nabla^{{\mathcal G}^D}_Y X, X).
 \end{eqnarray*}
Therefore, the condition $i_X{\rm d}^D(\flat_{{\mathcal G}^D}(X))=0$ is
alternatively written as
\begin{equation}\label{0closed}
{{\mathcal G}^D}(\nabla^{{\mathcal G}^D}_X X, Y)={{\mathcal
G}^D}(\nabla^{{\mathcal G}^D}_Y X, X)\; \hbox{  for every  } Y \in
\Gamma(\tau_D).
\end{equation}

We only need to check that both condition (ii) in Proposition \ref{propo1} and Theorem \ref{maintheorem} are equivalent.
If  we examine the Hamilton-Jacobi differential equation
\[
{\rm d}^{D}\left(\frac{1}{2}\mathcal G^D(X,X)+V\right) = 0,
\]
then for any $Y\in \Gamma(\tau_D)$ and $X\in \Gamma(\tau_D)$
satisfying $i_X{\rm d}^D(\flat_{{\mathcal G}^D}(X))=0$ we have that
\begin{eqnarray*}
0&=&{\rm d}^{D}(\frac{1}{2}{\mathcal G}^D(X,X)+V)(Y)\\
&=&\frac{1}{2}\rho_D(Y)({{\mathcal G}^D}(X,X))+\rho_D(Y)(V)\\
&=&{{\mathcal G}^D}(\nabla^{{\mathcal G}^D}_Y X, X)+\rho_D(Y)(V)\\
&=&{{\mathcal G}^D}(\nabla^{{\mathcal G}^D}_X X, Y)+{{\mathcal
G}^D}(\hbox{grad}_{{\mathcal G}^D}V, Y).
\end{eqnarray*}
In the last equality we have used condition (\ref{0closed}).
Therefore Equation (\ref{0h-j-p}) is written as
\[
 \nabla^{{\mathcal G}^D}_X X+\hbox{grad}_{{\mathcal G}^D}V=0\; .
\]
 (See also \cite{BMMP,LeMaMa}).
\end{proof}

From Theorem~\ref{maintheorem}, it is clear that we need to find
sections $X$ satisfying $i_X{\rm d}^D(\flat_{{\mathcal G}^D}(X))=0$.
The most simple candidate to be a solution is
$X=\hbox{grad}_{{\mathcal G}^D}f$ with $f: M\to \R$. Note that if $({\rm d}^D)^2(f)=0$, then $i_X{\rm d}^D(\flat_{{\mathcal G}^D}(X))=0$. After some
straightforward calculations, this condition is equivalent to:
\[
[\rho_D(Y), \rho_D(Z)](f)=(\rho_D\lcf Y, Z\rcf_D)(f)\; .
\]
This condition is always true if the bracket $\lcf \cdot, \cdot
\rcf_{D}$ satisfies the Jacobi identity, that is, if the pair $(\lcf
\cdot, \cdot \rcf_{D}, \rho_{D})$ is a \emph{Lie algebroid
structure} on the vector bundle $\tau_{D}: D \to M$ (see
\cite{LeMaMa}).

\section{Mechanical control systems and kinematic reductions}

Assume that the newtonian system determined by $(D, {\mathcal G}^D,
{\mathcal  F})$ also contains some input forces. We model this set
of input forces by a vector subbundle ${\mathcal D}^{(c)}$ of $D^*$.
Locally,  ${\mathcal D}^{(c)}=\hbox{span\;}\{\theta^1, \ldots,
\theta^k\}$, where $\theta^l\in\Gamma(\tau_{D^*})$, $1\leq l\leq k$.
Denote by ${\mathcal D}_{(c)}$ the vector subbundle of $D$ defined
by ${\mathcal D}_{(c)}=\sharp_{{\mathcal G}^D}({\mathcal D}
^{(c)})$. Therefore, locally ${\mathcal D}_{(c)}=\hbox{span\;}\{Y_1,
\ldots, Y_k\}$ where $Y_l=\sharp_{{\mathcal G}^D}\theta^l$,  $Y_l\in
\Gamma(\tau_D)$, $1\leq l\leq k$.  The vector fields $Y_1, \ldots,
Y_k$ are called the \emph{control sections} or \emph{input
sections}.

The equations of motion for a newtonian system with input
sections are as follows
\begin{equation}\label{noholonoma-1}
\nabla^{{\mathcal G}^D}_{\gamma(t)}\gamma(t) -{\mathcal
F}(\gamma(t))\in {\mathcal D}_{(c)}(\gamma(t)), \quad \forall \; t\in
I \subseteq \mathbb{R},
\end{equation}
where $\gamma: I\to D$ is a $\rho_D$-admissible curve.

In terms of the control sections, Equation (\ref{noholonoma-1}) can
be rewritten as follows:
\begin{equation}\label{noholonoma-1-1}
\nabla^{{\mathcal G}^D}_{\gamma(t)}\gamma(t)-{\mathcal F}(\gamma(t))
=\sum_{l=1}^k u^l(t) Y_l(\tau_D(\gamma(t)))
\end{equation}
for some $u: I\subseteq \R \to \R^k$, playing the role of
controls. The corresponding local equations are
\begin{equation*}\begin{array}{rcl}
\dot{x}^i&=&(\rho_D)^i_A y^A,\\
\dot{y}^C&=&-\Gamma^C_{AB}y^Ay^B+ {\mathcal F}^C(x,y)+\sum_{l=1}^k
u^l Y_l(x).
\end{array} \label{eqocontrol}
\end{equation*}
Note that if $Y \in \Gamma(\tau_{{\mathcal D}_{(c)}})$ then the integral curves of the vector 
field $\xi_{{\mathcal G}^D, {\mathcal F}} + Y^{\bf v}$ on $D$ are solutions of the previous equations, where
$Y^{\bf v}$ is the vertical lift of the section $Y$. 
\begin{remark}
{\rm It is also possible  to study the more realistic  case
when $U$ is a proper subset of $\R^k$. In this case, the
controls $u$ take value in a proper set of $\R^k$ (i.e.,
not all the linear combinations of controls are allowed).
Our procedure can be adapted to that particular control
set. But for geometrical clarity in this paper we only
consider the control distribution ${\mathcal D}_{(c)}$. }
\end{remark}

\begin{definition}
The 4-tuple $(D, {{\mathcal G}^D}, {\mathcal F}, {\mathcal
D}_{(c)})$ is called a \emph{mechanical control system on a skew-symmetric
algebroid}.
\end{definition}

%Restrict the bundle metric ${{\mathcal G}^D}$ to ${\mathcal D}_c$. Consequently, we obtain a new metric ${{\mathcal G}^{{\mathcal D}_c}}$ on the vector subbundle $\tau_{{\mathcal D}_c}: {\mathcal D}_c\to M$.
%Moreover, using the canonical inclusion  $i_{{\mathcal D}_c}: {\mathcal D}_c \to D$,
Consider  the orthogonal decomposition $D = {\mathcal D}_{(c)} \oplus
{\mathcal D}_{(c)}^{\perp}$ induced by the bundle metric ${{\mathcal G}^D}$ , with associated orthogonal projectors $P_{(c)}: D \to {\mathcal D}_{(c)}$
and $Q_{(c)}: D \to  {\mathcal D}_{(c)}^{\perp}$.

%, we can induce a skew-symmetric algebroid structure $({\mathcal D}, \lcf\;,\; \rcf_{{\mathcal D}_c}, \rho_{{\mathcal D}_c})$ on the vector bundle ${\mathcal D}_c\to M$ as in Example \ref{noholo-example}. This skew-symmetric algebroid structure is determined by
%\[
%\lcf X, Y\rcf_{{\mathcal D}_c}=P_c[i_{{\mathcal D}_c}(X), i_{{\mathcal D}_c}(Y)]\; ,\ \ \ \rho_{{\mathcal D}_c}(X)=i_{{\mathcal D}_c}(X)
%\]
%where $X, Y\in \Gamma (\tau_{{\mathcal D}_c})$.

%From the bundle metric ${{\mathcal G}^D}$ and the skew-symmetric structure
%$({\mathcal D}, \lcf\;,\; \rcf_{{\mathcal D}_c}, \rho_{{\mathcal D}_c})$ it is possible to construct the Levi-Civita connection as in Section (\ref{section2}).
The following proposition is a direct consequence of the definition
of mechanical control systems.
\begin{proposition}
An admissible curve  $\gamma: I\to D$ is solution of
Equation (\ref{noholonoma-1}) if and only if $\gamma: I\to
D$ satisfies  \[ Q_{(c)}\left(\nabla^{{\mathcal
G}^D}_{\gamma(t)}\gamma(t) -{\mathcal F}(\gamma(t))\right)
=0 \qquad \forall \; t\in I\subseteq \R.
\]
\end{proposition}

\subsection{Kinematic reduction of mechanical control systems on a
skew-symmetric algebroid}

 Now we introduce the
notion of a kinematic reduction (see \cite{bullolewis}).

Given an \emph{skew-symmetric algebroid structure} on the
vector bundle $\tau_{D}: D \to M$, we define  a
\emph{driftless system} as the set $(M,{\mathfrak D},
\overline{U})$, where ${\mathfrak D}$ is a vector subbundle
of $D$ locally spanned by $\{{X}_1, \ldots, {X}_{k'}\}$,
with $X_{\alpha}\in \Gamma(\tau_D)$, $1\leq \alpha\leq k'$,
and $\overline{U}\subset \R^{k'}$ is the set of admissible
controls. For a section $X=\sum_{\alpha=1}^{k'}
\overline{u}^\alpha X_\alpha \in \Gamma(\tau_D)$, remember that an
\textbf{integral curve of $X$} is a curve $\sigma
\colon I \subset \mathbb{R}\rightarrow M$ such that
\begin{equation}\label{drift}
\dot{\sigma}(t)=\sum_{\alpha=1}^{k'}
\overline{u}^{\alpha}(t)\rho_D({X}_{\alpha})(\sigma(t))\; ,
\end{equation}
where $(\overline{u}^{1}(t), \ldots,
\overline{u}^{k'}(t)) = (\bar{u}^1(\sigma(t)), \dots , \bar{u}^{k'}(\sigma(t))) \in \overline{U}$ for all $t$. It can
also be said that the pair $(\sigma, \overline{u})$ is a
solution to the driftless system.

Observe that for each pair $(\sigma(t), \overline{u}(t))$ we have
the curve $\gamma: I\longrightarrow {\mathfrak D}\subseteq D$
defined by
\[
\gamma(t)=\sum_{\alpha=1}^{k'}\overline{u}^{\alpha}(t){X}_{\alpha}(\sigma(t)).
\]
In the sequel, we will denote by $\tau_{{\mathfrak D}}$ the
restriction $(\tau_D)\Big|_{{\mathfrak D}}$.

\begin{definition}[Kinematic reduction]
Let
 $(D,  {\mathcal G}^D, {\mathcal F},  {\mathcal D}_{(c)})$ be a mechanical control system. A driftless system
$(M, {\mathfrak D}, \overline{U})$ is called a \emph{kinematic
reduction} of $(D,  {\mathcal G}^D, {\mathcal F},  {\mathcal
D}_{(c)})$ if
%\begin{enumerate}
%\item ${\mathfrak D}$ generates a vector subbundle of $D$, and
 for every solution $(\sigma(t), \overline{u}(t))$ of
(\ref{drift}) there exists a pair $(\gamma(t), u(t))$
solution of (\ref{noholonoma-1-1}), where
$\gamma(t)=\sum_{\alpha=1}^{k'}\overline{u}^{\alpha}(t){X}_{\alpha}(\sigma(t))$.
   \label{def:KinRed}
    \end{definition}

The rank of a kinematic reduction is the rank of the
distribution ${\mathfrak D}$. Rank-one kinematic reductions
are particularly interesting. A section $X$  of $\Gamma(\tau_D)$ is
called a \emph{decoupling section} if  the rank-one
kinematic system induced by ${\mathfrak D}= \hbox{span
}\{X\}$ is a kinematic reduction.

\begin{definition}[Kinematic controllability] A mechanical control system
 $(D,  {\mathcal G}^D, {\mathcal F},  {\mathcal D}_{(c)})$ is
 \emph{kinematically controllable} if it possesses decoupling sections $\{X_1,\dots,X_{k'}\}$ whose involutive closure has maximum
 rank.\label{dfn:KinControl}
\end{definition}

When a system is kinematically controllable, motion planning is
possible by using concatenations of integral curves of the
decoupling vector fields. Those curves must be reparametrized in
such a way that each segment begins and ends with zero velocity,
see~\cite{bullolewis} for more details.

We have the following adaptation of the results in Section 4.

\begin{proposition}\label{propo5}
 Let $(D, \lcf\; ,\rcf_D, \rho_D)$ be a skew-symmetric algebroid and consider a mechanical control problem determined by $(D,  {\mathcal G}^D, {\mathcal F}, {\mathcal D}_{(c)})$.
Consider a driftless system $(M, {\mathfrak D}, \overline{U})$. For
all $X\in \Gamma( \tau_{\mathfrak D})$, the following conditions are
equivalent:
\begin{enumerate}
\item if
$\sigma: I \longrightarrow M$ is an integral curve of $\rho_D(X)$,
that is,
\begin{equation*}\label{sigma-3}
\dot{\sigma}(t)=\rho_D(X)(\sigma(t)),
\end{equation*}
then $\gamma=X\circ \sigma: I \longrightarrow D$ is  an admissible curve solution of
\[
Q_{(c)}(\nabla^{{\mathcal G}^D}_{\gamma(t)}\gamma(t)
-{\mathcal F}(\gamma(t)))=0.
\]

\item $
Q_{(c)}\left(\nabla^{{\mathcal G}^D}_X X -{\mathcal
F}\circ X\right)=0$.
\end{enumerate}
\end{proposition}
%
%\begin{proof}
%Observe that locally if $X=X^a e_a$   then
%\begin{eqnarray*}
%&&\left(\nabla^{{\mathcal G}^D}_X X +\hbox{grad}_{{\mathcal
%G}^D}V\right) (\sigma(t))\\
%&&=
%\left(\Gamma^c_{ab}(\sigma(t)) \gamma^a(t) \gamma^b(t)+\frac{d}{dt}\gamma^c(t)+({{\mathcal G}^D})^{cb}(\sigma(t))(\rho_D)^i_b(\sigma(t))\frac{\partial V}{\partial x^i}\right) e_c(\sigma(t))
%\end{eqnarray*}
%where $\sigma$ is the curve defined on  (\ref{sigma-3}), $\gamma(t)=X(\sigma(t))$ and, in consequence, $\gamma(t)=\gamma^a(t)e_a(\sigma(t))=X^a(\sigma(t)) e_a(\sigma(t))$. From equation (\ref{equa1}) we directly deduce the equivalence of (i) and (ii).
%\end{proof}
%
%\begin{remark}
%{\rm In standard riemannian geometry (that is, $D=TM$ equipped with the standard lie bracket and $V\equiv 0$) the vector fields $X$ satisfying
%$\nabla^{{\mathcal G}^D}_X X$ are called auto-parallel vector fields and are obviously connected with the solutions of the geodesic equations}.
%\end{remark}
%
% Specializing Proposition \ref{propo1} to vector fields $X$ verifying an extra condition  $d^D(\flat_{\mathcal G}(X)) = 0$ we obtain a new expression of this Proposition. Indeed, the following Theorem give the classical expression of the Hamilton-Jacobi equation, proposed in \cite{LeMaMa}, corresponding to the hamiltonian function $h: D^*\longrightarrow \R$:
% \[
% h(\kappa, \kappa)={\mathcal G}^{D^*}(\kappa, \kappa)+V(\tau_{D^*}(\kappa))
% \]
% where $\kappa\in D^*$ and ${\mathcal G}^{D^*}: D^*\times_M D^*\longrightarrow \R$ is the induced bundle metric on the dual space.

Under extra assumptions, we have an alternative way to
write condition (ii) in Proposition \ref{propo5}. This
provides us with a new characterization of kinematic
reductions in terms of the affine connection of the given
mechanical control system. (See also
\cite{bullolewis,MuYa}).

\begin{proposition}
If $Q_{(c)}\left({\mathcal F}\right)=0$ and ${\mathfrak
D}=\hbox{span }\{ X_1, \ldots, X_{k'}\}$, then the
following conditions are equivalent:

\begin{enumerate}
\item
For all $X\in \Gamma(\tau_{\mathfrak D})$,
$Q_{(c)}\left(\nabla^{{\mathcal G}^D}_X X -{\mathcal
F}\circ X \right)=Q_{(c)}\left(\nabla^{{\mathcal G}^D}_X X
\right)=0 $.
\item
     For all $ \alpha,\beta, \gamma\in \{1, \ldots, k'\}$, $Q_{(c)}({X}_{\alpha})=0$,  $Q_{(c)} (\langle {X}_{\beta}: {X}_{\gamma}
     \rangle_{{\mathcal G}_D})=0$. In other words, for all $ \alpha,\beta, \gamma\in \{1, \ldots, k'\}$, \ ${X}_{\alpha}\in \Gamma(\tau_{{\mathcal D}_{(c)}})$
     and $\langle X_{\beta}: X_{\gamma}\rangle_{{\mathcal G}_{\mathcal D}} \in \Gamma(\tau_{{\mathcal D}_{(c)}})$.
     \item
$(M,{\mathfrak D},\overline{U})$ is a kinematic reduction of $(D,
{\mathcal G}^D, {\mathcal F}, {\mathcal D}_{(c)})$.
\end{enumerate} \label{prop:EquivKinematicReduction}
\end{proposition}
 \begin{proof}
 (i) $\Rightarrow$ (ii)
 Observe that for each $\alpha=1,\dots,k'$ we have that $fX_{\alpha}\in \Gamma(\tau_{\mathfrak D})$ for all $f\in C^{\infty}(M)$.
 If (i) holds, we have
 \[
 0=Q_{(c)}\left(\nabla^{{\mathcal G}^D}_{fX_{\alpha}} fX_{\alpha} \right)=f^2 Q_{(c)}\left(\nabla^{{\mathcal G}^D}_{X_{\alpha}} X_{\alpha}\right)+f \rho_D(X_{\alpha})(f) Q_{(c)}(X_{\alpha})\;
 \]
 for every $f\in {\mathcal C}^\infty(M)$. Then, taking
 suitable functions $f$, we obtain $Q_{(c)}(X_{\alpha})=0$ and $Q_{(c)}(\nabla^{{\mathcal G}^D}_{X_{\alpha}} X_{\alpha})=0$. Now, using the polarization identity we have that
 \[
\langle {X}_{\alpha}: {X}_{\beta}\rangle_{{\mathcal
G}_D}=\ds{\frac{1}{2}\left(\nabla^{{\mathcal
G}^D}_{X_{\alpha}+X_{\beta}}(X_{\alpha}+X_{\beta})-\nabla^{{\mathcal
G}^D}_{X_{\alpha}} X_{\alpha}-\nabla^{{\mathcal G}^D}_{X_{\beta}}
X_{\beta}\right)}\; .
 \]
 Therefore, $Q_{(c)} (\langle X_{\alpha}: X_{\beta}\rangle_{{\mathcal G}_D})=0$.

(ii) $\Rightarrow$ (iii) By definition ${\mathfrak D}
\subseteq D$. As (ii) is true, then statement (ii) in
Proposition~\ref{propo5} is satisfied for every $X\in
\Gamma(\tau_D)$. Hence statement (i) in
Proposition~\ref{propo5} is true and we can conclude that
$(M,{\mathfrak D},\overline{U})$ is a kinematic reduction
of $(D, {\mathcal G}^D, {\mathcal F}, {\mathcal D}_{(c)})$
according to Definition~\ref{def:KinRed}.

(iii) $\Rightarrow$ (i) As $(M,{\mathfrak D},\overline{U})$ is a
kinematic reduction of $(D, {\mathcal G}^D, {\mathcal F})$,
statement (i) in Proposition~\ref{propo5} is satisfied. Hence the
result follows.
 \end{proof}

A straightforward corollary of
Proposition~\ref{prop:EquivKinematicReduction} is the
following one:

\begin{corollary}
If $Q_{(c)}\left({\mathcal F}\right)=0$ and ${\mathfrak
D}=\hbox{span }\{ X_1\}$, then the following conditions are
equivalent:

\begin{enumerate}
\item
$Q_{(c)}\left(\nabla^{{\mathcal G}^D}_{X_1} X_1 -{\mathcal
F}\circ X_1\right)=Q_{(c)}\left(\nabla^{{\mathcal G}^D}_{X_1} X_1
\right)=0 $.
     \item
$(M,{\mathfrak D},\overline{U})$ is a rank-one kinematic
reduction of $(D, {\mathcal G}^D, {\mathcal F}, {\mathcal
D}_{(c)})$.
\end{enumerate}\label{Corol:Decoupling}
\end{corollary}

In the sequel assume that ${\mathcal F}$ comes from a
potential function $V: M\to \R$, that is, ${\mathcal
F}=-{\rm grad}_{{\mathcal G}^D}V\circ \tau_D$. Then
Theorem~\ref{maintheorem} can be adapted to control systems
as follows.

\begin{theorem}\label{maintheorem-5}
Let $(D, \lcf\; ,\rcf_D, \rho_D)$ be a skew-symmetric algebroid and
consider a control  problem determined by $(D, {\mathcal G}^D, V,
{\mathcal D}_{(c)})$.

Take a section $X\in \Gamma(\tau_{D})$ such that $i_X{\rm
d}^D(\flat_{{\mathcal G}^D}(X))(Y)= 0$ for all $Y\in \Gamma(\tau_{{\mathcal
D}_{(c)}^{\perp}})$. Under this hypothesis, the following conditions
are equivalent:
\begin{enumerate}
\item
If $\sigma: I \longrightarrow M$ is an integral curve of
$\rho_D(X)$, that is,
\begin{equation*}\label{sigma-3-mainthm-5}
\dot{\sigma}(t)=\rho_D(X)(\sigma(t)),
\end{equation*}
then $\gamma=X\circ \sigma: I \longrightarrow D$ is  an admissible
curve solution of
\[
{Q}_{(c)}(\nabla^{{\mathcal G}^D}_{\gamma(t)}\gamma(t) +
\hbox{grad}_{{\mathcal G}^D}V(\tau_D(\gamma(t))))=0.
\]
\item
$X$ satisfies \emph{the Hamilton-Jacobi differential equation}
\begin{equation}\label{eq:thm5h-j-p}
{\rm d}^{D}\left(\frac{1}{2}\mathcal G^D(X,X)+V\right) (Y) = 0
\hbox{ for all } Y\in \Gamma(\tau_{{\mathcal D}_{(c)}^{\perp}})\; .
\end{equation}
\end{enumerate}
\end{theorem}

\begin{proof}

Following similar arguments that in the proof of Theorem
\ref{maintheorem} it is easy to prove that the condition $i_X{\rm
d}^{D}(\flat_{{\mathcal G}^D}(X))(Y)= 0$ for all $Y\in \Gamma(\tau_{{\mathcal
D}_{(c)}^{\perp}})$ is equivalent to
\[
{\mathcal G}^D(\nabla^{{\mathcal G}^D}_X X, Y)={\mathcal G}^D(\nabla^{{\mathcal G}^D}_Y X, X)\; .
\]

Finally,  observe now  that
\begin{eqnarray*}
0&=&{\rm d}^{D}\left(\frac{1}{2}{\mathcal G}^D(X,X)+V\right)(Y)\\
&=&{{\mathcal G}^D}(\nabla^{{\mathcal G}^D}_Y X, X)+\rho_D(Y)(V)\\
&=&{{\mathcal G}^D}(\nabla^{{\mathcal G}^D}_X X, Y)+{{\mathcal G}^D}(\hbox{grad}_{{\mathcal G}^D}V, Y)\; ,
\end{eqnarray*}
The last expression is equivalent to the equation
\[
Q_{(c)}(\nabla^{{\mathcal G}^D}_X X+\hbox{grad}_{{\mathcal G}^D}V)=0\; .
\]
\end{proof}

Hence, we have just extended the notion of decoupling sections for mechanical control systems with nonzero potential
since Theorem~\ref{maintheorem-5} gives sufficient and necessary conditions to have a kinematic reduction of
such a mechanical control system.

\begin{remark} Note that the condition $i_X {\rm d}^D(\flat_{{\mathcal G}^D}(X))(Y)=0$ in the hypothesis of Theorem~\ref{maintheorem-5} is
${\mathcal C}^\infty(M)$-linear in $Y\in {\mathcal D}_{(c)}^\perp$.
Hence only for a basis of vector fields in ${\mathcal
D}_{(c)}^\perp$ the condition must be checked. However that same
condition is not ${\mathcal C}^\infty(M)$-linear in $X\in
\Gamma(\tau_{D})$.
\end{remark}

This theorem plays a key role to define in the future the geometric notion of motion planning for mechanical control systems
with non-zero potentials, only known so far for zero potentials \cite{bullolewis}. In that sense it will be useful to have some notion
of reparametrization of integral curves of sections $X\in
\Gamma(\tau_{D})$.

\begin{proposition}
 Let $X$ be a nonzero section in $\Gamma(\tau_{{\mathcal
D}_{(c)}})$ such that for all $Y\in \Gamma(\tau_{{\mathcal
D}_{(c)}^{\perp}})$, $i_X{\rm
d}^D(\flat_{{\mathcal G}^D}(X))(Y)= 0$. For a nonzero function $f$ in ${\mathcal C}^\infty(M)$, the section $fX\in \Gamma(\tau_{{\mathcal
D}_{(c)}})$ satisfies
$i_{fX}{\rm
d}^D(\flat_{{\mathcal G}^D}(fX))(Y)= 0$ for all $Y\in \Gamma(\tau_{{\mathcal
D}_{(c)}^{\perp}})$ if and only if $\rho_D(Y)(f)=0$ for all $Y\in \Gamma(\tau_{{\mathcal
D}_{(c)}^{\perp}})$.

\begin{proof}
 Let us rewrite the condition $i_{fX}{\rm
d}^D(\flat_{{\mathcal G}^D}(fX))(Y)= 0$ as follows
\begin{eqnarray*}
 0&=&{\mathcal G}^D(\nabla^{{\mathcal G}^D}_{fX} fX, Y)-{\mathcal G}^D(\nabla^{{\mathcal G}^D}_Y fX, fX)\\
&=& f^2\left({\mathcal G}^D(\nabla^{{\mathcal G}^D}_X X, Y)-{\mathcal G}^D(\nabla^{{\mathcal G}^D}_Y X, X)\right)
+ f\rho_{D}(X)(f)  {\mathcal G}^D(X,Y)\\
&-&f\rho_D(Y)(f) {\mathcal G}^D(X,X)=-f\rho_D(Y)(f) {\mathcal G}^D(X,X)
\end{eqnarray*}
because  $X\in \Gamma(\tau_{{\mathcal
D}_{(c)}})$. From here, the equivalence is straightforward.

\end{proof}
\label{prop:Cond-Contr-fX}
\end{proposition}

As a consequence of Proposition~\ref{prop:Cond-Contr-fX}, Theorem~\ref{maintheorem-5} can also be written for such a $fX$.

From Theorem~\ref{maintheorem-5}, we establish a connection between
decoupling sections in the sense defined in \cite{bullolewis} and the solutions to Hamilton-Jacobi differential equation.

\begin{corollary}
Let $(D, \lcf\; ,\rcf_D, \rho_D)$ be a skew-symmetric
algebroid and consider a control  problem determined by
$(D, {\mathcal G}^D, V, {\mathcal D}_{(c)})$ where $V$
satisfies that ${\rm grad}_{{\mathcal G}^D}V\in \Gamma(\tau_{{\mathcal
D}_{(c)}})$. If a section $X\in \Gamma(\tau_{{\mathcal
D}_{(c)}})$ is a decoupling section, then the following
conditions are equivalent:
\begin{enumerate}
\item \label{Cond1CorollaryDecoupling} $i_X{\rm
d}^D(\flat_{{\mathcal G}^D}(X)) )(Y)= 0$ for all $Y\in
\Gamma(\tau_{{\mathcal D}_{(c)}^{\perp}})$,
\item \label{eq:DecouplingCondition} $0={\mathcal G}^D(\nabla^{{\mathcal G}^D}_Y X,X)
\quad \mbox{for all } Y\in \Gamma(\tau_{{\mathcal D}_{(c)}^{\perp}})$,
\item \label{eq:DecouplingHJ}  $X$ satisfies \emph{the Hamilton-Jacobi differential equation}
\begin{equation*}\label{eq:h-j-pCorollary}
{\rm d}^{D}\left(\frac{1}{2}\mathcal G^D(X,X)+V\right) (Y)
= 0 \hbox{ for all } Y\in \Gamma(\tau_{{\mathcal D}_{(c)}^{\perp}})\; .
\end{equation*}
\end{enumerate}\label{Corol:DecouplingVZero}
\end{corollary}

\begin{proof}
The condition $i_X{\rm d}^D(\flat_{{\mathcal G}^D}(X))(Y)= 0$ for
all $Y\in \Gamma(\tau_{{\mathcal D}_{(c)}^{\perp}})$ can be rewritten as
\[{\mathcal G}^D(\nabla^{{\mathcal G}^D}_X X,Y)={\mathcal G}^D(\nabla^{{\mathcal G}^D}_Y X,X)\]
for all $Y\in \Gamma(\tau_{{\mathcal D}_{(c)}^{\perp}})$.

As $X$ is a decoupling section, both $X$ and
$\nabla^{{\mathcal G}^D}_X X\in \Gamma(\tau_{{\mathcal D}_{(c)}})$
by Corollary~\ref{Corol:Decoupling}. Hence the left-hand
side of the above equality is zero
and~\ref{Cond1CorollaryDecoupling} is equivalent
to~\ref{eq:DecouplingCondition}.

Condition~\ref{eq:h-j-pCorollary} can be rewritten as follows:
\begin{align*} 0&=\; {\rm d}^{D}\left(\frac{1}{2}\mathcal G^D(X,X)+V\right) (Y)
=\frac{1}{2}\rho_D(Y)\mathcal G^D(X,X)+\rho_D(Y)(V)\\&=\;
\mathcal G^D(\nabla^{\mathcal G^D}_YX,X)+\mathcal G^D({\rm
grad}_{{\mathcal G}^D}V,Y),
\end{align*}
for all $Y\in \Gamma(\tau_{{\mathcal D}_{(c)}^{\perp}})$. Hence the
equivalence between~\ref{eq:DecouplingCondition}
and~\ref{eq:h-j-pCorollary} is clear because of the
hypothesis for $V$.
\end{proof}

Note that the property of being a decoupling section is preserved by ${\mathcal C}^\infty(M)$-multiplication.
However, the condition $i_X{\rm
d}^D(\flat_{{\mathcal G}^D}(X))(Y)= 0$ for all $Y\in \Gamma(\tau_{{\mathcal
D}_{(c)}^{\perp}})$ is not ${\mathcal C}^\infty(M)$-linear on $X$.  Proposition~\ref{prop:Cond-Contr-fX} has already
characterized those functions $f$ such that $i_{fX}{\rm
d}^D(\flat_{{\mathcal G}^D}(fX))(Y)= 0$ for all $Y\in \Gamma(\tau_{{\mathcal
D}_{(c)}^{\perp}})$.

It might also be the case that a decoupling section
$X\in \Gamma(\tau_{{\mathcal D}_{(c)}})$ does not
satisfy~\ref{eq:DecouplingCondition} in
Corollary~\ref{Corol:DecouplingVZero}, but there might exist
functions $f\in {\mathcal C}^\infty(M)$ such that $fX$
satisfies~\ref{eq:DecouplingCondition}. Let $Y\in \Gamma(\tau_{{\mathcal
D}_{(c)}^{\perp}})$,

\begin{align}{\mathcal
G}^D(\nabla^{{\mathcal G}^D}_Y fX,fX)&=\; f{\mathcal
G}^D(\rho_D(Y)(f)X+f\nabla^{{\mathcal G}^D}_Y X,X) \nonumber \\
&=\;f\rho_D(Y)(f){\mathcal G}^D(X,X)+f^2{\mathcal G}^D(\nabla^{{\mathcal
G}^D}_Y X,X). \label{eq:DecouplingWithF}
\end{align}
Condition (ii) in Corollary~\ref{Corol:DecouplingVZero} would be
satisfied if the following partial differential equation for $f$ has
solutions:
\begin{equation*} f\rho_D(Y)(f){\mathcal G}^D(X,X)+f^2{\mathcal G}^D(\nabla^{{\mathcal
G}^D}_Y X,X)=0
\end{equation*}
for all $Y\in \Gamma(\tau_{{\mathcal D}_{(c)}^{\perp}})$. The chances to find a
solution depend on the particular examples under study, see Section~\ref{Sec:example}.

% The following result states when all ${\mathcal C}^\infty(M)$-multiples of a decoupling section $X$
% satisfying~\ref{eq:DecouplingCondition} in
% Corollary~\ref{Corol:DecouplingVZero} also
% satisfy~\ref{eq:DecouplingCondition}.
%
% \begin{corollary}
% Let $(D, \lcf\; ,\rcf_D, \rho_D)$ be a skew-symmetric
% algebroid and consider a control  problem determined by
% $(D, {\mathcal G}^D, V, {\mathcal D}_{(c)})$  where $V$
% satisfies that ${\rm grad}_{{\mathcal G}^D}V\in {\mathcal
% D}_{(c)}^\perp$. If a section $X\in \Gamma(\tau_{D})$ is
% decoupling and solution to $0={\mathcal
% G}^D(\nabla^{{\mathcal G}^D}_Y X,X)$ for all $Y\in
% {\mathcal D}_{(c)}^{\perp}$, then for a nonzero function
% $f\in{\mathcal C}^\infty(M)$, a section $fX\in
% \Gamma(\tau_{D})$ satisfies the same condition if and only
% if $\rho_D(Y)(f)=0$ for all $Y\in {\mathcal
% D}_{(c)}^{\perp}$. \label{Corol:CondFX}
% \end{corollary}
% \begin{proof}
% It follows from~\eqref{eq:DecouplingWithF} since
% \begin{equation*}
% 0={\mathcal G}^D(\nabla^{{\mathcal G}^D}_Y
% fX,fX)=f\rho_D(Y)(f){\mathcal G}^D(X,X).
% \end{equation*}
% \end{proof}

\subsection{Maximally reducible systems}

If $(M,{\mathfrak D},\overline{U})$ is a kinematic reduction of $(D,
{\mathcal G}^D, {\mathcal F}, {\mathcal D}_{(c)})$, then any
solution of~\eqref{drift} can be followed by a solution
of~\eqref{noholonoma-1}. In this section we consider when the
converse is also possible in such a way that we can talk about
``equivalence" of controlled trajectories as mentioned in the
introduction.

\begin{definition}[Maximal reducibility]
A mechanical control system  $(D,  {\mathcal G}^D,
{\mathcal F}, {\mathcal D}_{(c)})$ being ${\mathcal
D}_{(c)}$ a subbundle of $D$ is \emph{maximally reducible}
to a driftless system $(M, {\mathfrak D}, \overline{U})$ if
\begin{enumerate}
\item $(M, {\mathfrak
D}, \overline{U})$ is a kinematic reduction of $(D,  {\mathcal G}^D,
{\mathcal F}, {\mathcal D}_{(c)})$, and
\item for every solution $(\gamma(t), u(t))$ of Equations (\ref{noholonoma-1-1}) satisfying $\gamma(0)\in \mathfrak{D}_{(\tau_D \circ \gamma)(0)}$,
there exists a control $\overline{u}\in \overline{U}$ such that
$(\tau_D\circ \gamma(t), \overline{u}(t))$ is a solution of
Equation~\eqref{drift}.
    \end{enumerate}\label{def:MaxRed}
    \end{definition}

The characterization of maximally reducible mechanical control systems defined on skew-symmetric algebroids is given by the following result. This a a generalization of the notion of being maximally reducible systems proved in~\cite{bullolewis} on Riemannian manifolds.

\begin{theorem} Let $(D,  {\mathcal G}^D, {\mathcal F},
{\mathcal D}_{(c)})$ be a mechanical control system  such that 
${\mathcal D}_{(c)}$ has locally constant rank  and
${\mathcal F}=0$. This mechanical control system is
maximally reducible to a driftless system $(M, {\mathfrak
D}, \overline{U})$ if and only if the following two
conditions hold:
\begin{enumerate}
\item ${\mathfrak
D}={\mathcal D}_{(c)}$,
\item ${\mathcal D}_{(c)}$ is geodesically invariant, that is,
 ${\rm Sym}^{(\infty)}({\mathcal D}_{(c)})={\mathcal D}_{(c)}$,
\end{enumerate}
where ${\rm Sym}^{(\infty)}({\mathcal D}_{(c)})$ is the
smallest distribution containing ${\mathcal D}_{(c)}$ and
closed under the symmetric product $\langle \cdot \; \colon
\cdot \rangle_{{\mathcal G}^D}$.\label{thm:MaxRed}
\end{theorem}

\begin{proof}
It follows the same lines as the proof of Theorem~8.27
in~\cite{bullolewis}.

Suppose that the rank of $\mathfrak D \subseteq D$ is $k'$. 
Note that $X\in \Gamma(\tau_{D})$ is 
a section of $\mathfrak D$ if and only if the vertical lift of $X$ restricted to $\mathfrak D$ 
 is tangent to $\mathfrak D$.  

Assume that $(D,  {\mathcal G}^D, 0, {\mathcal D}_{(c)})$
is maximally reducible to a driftless system $(M, {\mathfrak D}, \overline{U})$. 

If $\gamma: I \subseteq \mathbb{R} \rightarrow D$ is a geodesic for the bundle metric ${\mathcal G}^D$ 
 with initial condition  $\gamma(0) \in {\mathfrak D}_{\sigma(0)}$, then $\gamma$ is a solution of 
Equations \eqref{noholonoma-1-1} with zero controls. Thus, by hypothesis, there exist controls
$\bar{u}^\alpha: I \rightarrow \bar{U}$ such that $\gamma(t) = \displaystyle \sum_{\alpha = 1}^{k'} \bar{u}^\alpha(t) X_\alpha(\sigma(t))$, where $\sigma = \tau_D \circ \gamma: I \subseteq
\mathbb{R} \rightarrow M$.
Therefore, $\gamma(I) \subseteq {\mathfrak D}$ and we can conclude that $\mathfrak D$ is geodesically invariant.

It remains to prove that $\mathfrak D= {\mathcal D}_{(c)}$.

Remember that ${\mathcal D}_{(c)}={\rm span}\{Y_1,\dots, Y_k\}$. Then it is sufficient to prove that $Y_s$
is a section of $\mathfrak D$ for $s=1,\dots,k$ in order to obtain $ {\mathcal D}_{(c)}\subset \mathfrak D$. Having in mind that $\xi_{{\mathcal G}^{D}}$ is the geodesic spray on $D$ and 
$Y_s^{\bf v}$ is the vertical lift of $Y_s \in \Gamma(\tau_{{\mathcal D}_{(c)}})$ to ${\mathcal D}_{(c)}$ , it is clear that the integral curves of $\xi_{{\mathcal G}^{D}} + Y_s^{\bf v}$ are solutions 
of Equations \eqref{noholonoma-1-1}. By assumption, $\xi_{{\mathcal G}^{D}} + Y_s^{\bf v}$ restricted to 
 $\mathfrak D$  is tangent to $\mathfrak D$. On the other hand, $\mathfrak D$ is geodesically invariant, what implies that $\xi_{{\mathcal G}^{D}}$ restricted to $\mathfrak D$ is tangent to $\mathfrak D$.
Hence, $(Y_s^{\bf v})_{| \mathfrak D}$ is tangent to $\mathfrak D$ and $Y_s$ is a section of $\mathfrak D$.

Next, we will see that $\mathfrak D \subseteq {\mathcal D}_{(c)}$.
Let $a$ be a vector in $\mathfrak D$. We consider the geodesic $\tilde{\gamma}: \tilde{I} \subseteq \mathbb{R} \rightarrow
D$ with initial condition $\tilde{\gamma}(0) = a$. Hence, $\tilde{\gamma}(\tilde{I}) \subseteq \mathfrak D$ because  $\mathfrak D$ is geodesically invariant. Now, take the curve $\gamma: I \subseteq \mathbb{R} \rightarrow \mathfrak D$ given by
\[
\gamma(t) = t \tilde{\gamma}\left(\frac{t^2}{2}\right), \; \; \; \forall t \in I.
\]
Note that $\tau_D\circ \gamma= \tau_D \circ \tilde{\gamma} \circ \tau$, where $\tau\colon \mathbb{R} \rightarrow \mathbb{R}$ is the map defined by $\tau(t)=t^2/2$, and $\gamma(0) = 0$. From (\ref{3.0}) it is straightforward to prove that
\[
\left(\nabla^{{\mathcal G}^D}_{\gamma(0)}\gamma\right)(0) = a .
\]
On the other hand, define a curve  $\tilde{\sigma} $ on $M$ given by $\tilde{\sigma} = \tau_D \circ \tilde{\gamma}$ such that
\[
\gamma(t) = \displaystyle \sum_{\alpha = 1}^{k'} \bar{u}^{\alpha}(t) X_{\alpha}((\tau_D \circ \gamma)(t)) =  \displaystyle \sum_{\alpha = 1}^{k'} \bar{u}^{\alpha}(t) X_{\alpha}((\tilde{\sigma}\circ \tau)(t)).
\]
Then, it is clear that 
\[
\displaystyle \sum_{\alpha = 1}^{k'} \bar{u}^{\alpha}(t) \rho_{D}(X_{\alpha})((\tilde{\sigma} \circ \tau)(t)) = \rho_D(\gamma(t)) = \frac{d}{dt}(\tilde{\sigma} \circ \tau) = \frac{d}{dt}(\tau_D \circ \gamma),
\]
what implies that $(\tilde{\sigma} \circ \tau, \bar{u})$ is a solution of Equations~\eqref{drift}. Therefore, there exist controls $u$ such that $(\gamma, u)$ is a solution of Equations (\ref{noholonoma-1-1}). In particular,
\[
a = \nabla^{{\mathcal G}^{D}}_{\gamma(0)}\gamma(0) \in {\mathcal D}_{(c)}.
\]

The other implication is proved as follows: First, we prove that $(M,\mathfrak{D},\overline{U})$ is a kinematic reduction of 
$(D,  {\mathcal G}^D, {\mathcal F},
{\mathcal D}_{(c)})$. Let $\sigma\colon I \subseteq \mathbb{R} \rightarrow M$ be a curve such that it satisfies  $\dot{\sigma}(t)=\sum_{\alpha=1}^{k'}\overline{u}^\alpha(t)\rho_D(X_\alpha)(\sigma(t))$, i.e. Equation~\eqref{drift}. Take $\gamma(t)= \sum_{\alpha=1}^{k'}\overline{u}^\alpha(t)X_\alpha (\sigma(t))$ and note that $\gamma(I)\subseteq \mathfrak{D}$ and $\tau_D \circ \gamma=\sigma$. Moreover,
\begin{equation*}
\ds{\frac{{\rm d}}{{\rm d}t} (\tau_D \circ \gamma)=\dot{\sigma}=\sum_{\alpha=1}^{k'}\overline{u}^\alpha(t)\rho_D(X_\alpha)(\sigma(t))=\rho_D(\gamma(t)).}
\end{equation*}
Thus $\gamma$ is $\rho_D$-admissible. 

Now we have to prove that $\nabla^{{\mathcal G}^D}_{\gamma(t)} \gamma(t)\in {\mathcal D}_{(c)}(\sigma(t))$. By hypothesis, $\mathfrak{D}={\mathcal D}_{(c)}$ and ${\mathcal D}_{(c)}$ is geodesically invariant. Then by Theorem~\ref{Thm:GeodInv} we conclude that $\nabla^{{\mathcal G}^D}_{\gamma(t)} \gamma(t)\in {\mathcal D}_{(c)}(\sigma(t))$ for all $t\in I$. Hence, there exist controls such that $\gamma$ is a solution of Equation~\eqref{noholonoma-1-1}.

It remains to prove condition $(ii)$ in Definition~\ref{def:MaxRed}. Let $(\gamma,u)\colon I \rightarrow D\times U $ be a solution to Equation~\eqref{drift}, then $Q_{(c)}(\nabla^{{\mathcal G}^D}_{\gamma(t)} \gamma(t))=0$, that is, 
\begin{equation*}
\nabla^{{\mathcal G}^D}_{\gamma(t)} \gamma(t)\in {\mathcal D}_{(c)}(\sigma(t))=\mathfrak{D}(\sigma(t))
\end{equation*}
where $\sigma=\tau_D \circ \gamma$.

By assumption $\gamma(0)\in \mathfrak{D}(\sigma(0))$ and $\gamma$ is an integral curve of $\xi_{{\mathcal G}^D}+X^{\bf v}$, where $\xi_{{\mathcal G}^D}$ is the geodesic spray associated with ${\mathcal G}^D$ and $X\in \Gamma(\tau_{\mathfrak{D}})$.

Note that $X^{\bf v}_{\left.\right|_{\mathfrak{D}}}$ is tangent to $\mathfrak{D}$ and $(\xi_{{\mathcal G}^D})_{\left.\right|_{\mathfrak{D}}}$ is also tangent to $\mathfrak{D}$ because $\mathfrak{D}$ is geodesically invariant. Hence, if $\gamma(0)\in \mathfrak{D}(\sigma(0))$, then the integral curve of $\xi_{{\mathcal G}^D}+X^{\bf v}$ with initial condition $\gamma (0)$ is entirely contained in $\mathfrak{D}$, that is, $\gamma(t)\in \mathfrak{D}(\sigma(t))$ for all $t\in I$. 

Then there exist $\overline{u}\colon I \rightarrow \overline{U}$ such that $\gamma(t)=\sum_{\alpha=1}^{k'} \overline{u}^\alpha(t) X_\alpha(\sigma(t))$ and \begin{equation*}\rho_D(\gamma(t))= \sum_{\alpha=1}^{k'} \overline{u}^\alpha(t) \rho_D(X_\alpha)(\sigma(t)).\end{equation*} Using the fact that $\gamma$ is $\rho_D$-admissible it immediately follows that $\dot{\sigma}(t)=\sum_{\alpha=1}^{k'} \overline{u}^\alpha(t) \rho_D(X_\alpha)(\sigma(t))$ for all $t\in I$.
\end{proof}

Let us provide some results to find decoupling sections for the
mechanical control systems under consideration.

\begin{corollary} If a mechanical control system  $(D,  {\mathcal G}^D, {\mathcal F},
{\mathcal D}_{(c)})$ with ${\mathcal F}=0$ is
maximally reducible, then all control sections are
decoupling.
\end{corollary}
\begin{proof} This follows immediately from Theorem~\ref{thm:MaxRed}
and the definition of decoupling sections.
\end{proof}

The converse is not necessarily true. But it is true when ${\mathcal D}_{(c)}$ has locally constant rank
equal to one.

\begin{corollary} Let $(D,  {\mathcal G}^D, {\mathcal F},
{\mathcal D}_{(c)})$ be a mechanical control system where
${\mathcal D}_{(c)}$ has rank one and ${\mathcal F}=0$. The following statements are equivalent:
\begin{enumerate}
\item There exist decoupling sections for  $(D,  {\mathcal G}^D, {\mathcal F},
{\mathcal D}_{(c)})$.
\item The mechanical control system  $(D,  {\mathcal G}^D, {\mathcal F},
{\mathcal D}_{(c)})$ is maximally
reducible to a driftless system defined by decoupling sections.
\end{enumerate}
\end{corollary}

\begin{proof} It is straightforward from Theorem~\ref{thm:MaxRed}.
\end{proof}

\subsection{Examples}\label{Sec:example}

The first two examples are specific cases of
Example~\ref{example:DequalTM}. As is shown, the example in
Section~\ref{Sec:ExPlanarBody} is not maximally reducible
but admits rank-one kinematic reductions. Particular
solutions to Hamilton-Jacobi differential equation are
found using Theorem~\ref{maintheorem-5} and
Corollary~\ref{Corol:DecouplingVZero}. In
Section~\ref{Sec:RoboticLeg} we compute particular
solutions to Hamilton-Jacobi differential equation for a
maximally reducible system. The snakeboard described in
Section~\ref{Sec:Snakeboard} has nonholonomic constraints.
Hence the use of skew-symmetric algebroids to find
solutions to Hamilton-Jacobi differential equations is very
natural.

\subsubsection{Planar rigid body with a variable-direction
thruster} \label{Sec:ExPlanarBody}

We refer to~\cite[Section 7.4.2]{bullolewis} for a detailed description of
the system. The configuration space is
$M=\mathbb{R}^2\times \mathbb{S}^1$. Consider local
coordinates $(x,y,\theta)$. Here the distribution $D$ is
the entire tangent space $TM$ and $V=0$. We are in the case
explained in Example~\ref{example:DequalTM}.

The riemannian metric is
\[{\mathcal G}^{TM}=J {\rm d} \theta \otimes {\rm d} \theta + m({\rm
d} x \otimes {\rm d}x + {\rm d} y \otimes {\rm d}y).\]

The control vector fields in ${\mathcal D}_{(c)}$ are
\[Y_1=\frac{\cos \theta}{m} \, \frac{\partial}{\partial x} +\frac{\sin \theta}{m} \, \frac{\partial}{\partial
y}, \quad Y_2=-\frac{\sin \theta}{m} \, \frac{\partial}{\partial
x}+\frac{\cos \theta}{m} \, \frac{\partial}{\partial y}
-\frac{h}{J}\, \frac{\partial}{\partial \theta}.\] In this example
\begin{equation*}{\mathcal D}_{(c)}^{\perp}={\rm span}_{{\mathcal
C}^\infty}\left\{Y_3\colon =\ds{-\sin
\theta\frac{\partial}{\partial x}+\cos
\theta\frac{\partial}{\partial y}+\frac{1}{h}\,
\frac{\partial}{\partial \theta}}\right\}.\end{equation*}
Note that $\{Y_1,Y_2,Y_3\}$ is a ${\mathcal
G}^{TM}$-orthogonal basis. The skew-symmetric algebroid
structure is defined as follows $D=TM$, $\lcf \cdot,
\rcf_{TM}$ is the usual Lie bracket and $\rho_{TM}={\rm
Id}_{TM}$. In our adapted basis we have
\begin{align*}
\lcf Y_1,Y_2\rcf_{TM}=[Y_1,Y_2]&=\;
\ds{\frac{h}{J+mh^2}Y_2+\frac{h^3}{(J+mh^2)J}Y_3},\\
\lcf Y_1,Y_3\rcf_{TM}=[Y_1,Y_3]&=\;
\ds{-\frac{J}{h(J+mh^2)}Y_2-\frac{h}{J+mh^2}Y_3},\\
\lcf Y_2,Y_3\rcf_{TM}=[Y_2,Y_3]&=\; \ds{\frac{mh^2+J}{hJ}Y_1}.
\end{align*}
The non-zero Christoffel symbols
for the associated Levi-Civita $\nabla^{{\mathcal
G}^{TM}}$ connection are
\begin{equation*}\begin{array}{lll}
\Gamma^1_{22}=\ds{\frac{h}{J}}, &
\Gamma^1_{23}=\ds{\frac{mh}{J}},&
\Gamma^1_{32}=-\ds{\frac{1}{h}},
\\ \Gamma^1_{33}=-\ds{\frac{m}{h}},& \Gamma^2_{21}=-\ds{\frac{h}{J+mh^2}}, &\Gamma^2_{31}=\ds{\frac{J}{h(J+mh^2)}},\\
\Gamma^3_{21}=-\ds{\frac{h^3}{J(J+mh^2)}},&
\Gamma^3_{31}=\ds{\frac{h}{J+mh^2}}.
&\end{array}\end{equation*} Let us compute the symmetric
products of the control vector fields
\begin{eqnarray*}
<Y_1\colon Y_1>_{{\mathcal G}^D}&=0, & <Y_1\colon Y_2>_{{\mathcal
G}^D}=\Gamma^2_{21}Y_2+\Gamma^3_{21}Y_3,\\
<Y_2\colon Y_2>_{{\mathcal G}^D} &=2\Gamma^1_{22}Y_1.&
\end{eqnarray*}
By Corollary~\ref{Corol:Decoupling}, $Y_1$ and $Y_2$ are
decoupling vector fields. However, ${\mathcal D}_{(c)}={\rm
span}\{Y_1, Y_2\}$ is not geodesically invariant because
${\rm Sym}^{(1)}{\mathcal D}_{(c)} \nsubseteq {\mathcal
D}_{(c)}$. Then according to Theorem~\ref{thm:MaxRed} the
mechanical control system is not maximally reducible to a
driftless control system. As the involutive closure of the
decoupling vector fields $Y_1$ and $Y_2$ has maximum rank,
the system is kinematic controllable, see
Definition~\ref{dfn:KinControl}. Then the motion planning
is feasible.

Let us see if $Y_1$ and $Y_2$ satisfy
condition~\ref{eq:DecouplingCondition} in
Corollary~\ref{Corol:DecouplingVZero}, that is,
$0={\mathcal G}^{TM}(\nabla^{{\mathcal G}^{TM}}_Y X,X)$ for
all  $Y\in {\mathcal D}_{(c)}^{\perp}={\rm span}\{Y_3\}$.
\begin{align*}
{\mathcal G}^{TM}(\nabla^{{\mathcal G}^{TM}}_{Y_3} Y_1,Y_1)&\; =
\Gamma^2_{31}{\mathcal G}^{TM}(Y_2,Y_1)+\Gamma^3_{31}{\mathcal
G}^{TM}(Y_3,Y_1)=0,\\
{\mathcal G}^{TM}(\nabla^{{\mathcal G}^{TM}}_{Y_3} Y_2,Y_2)&\; =
\Gamma^1_{32}{\mathcal G}^{TM}(Y_1,Y_2)=0\\
\end{align*}
since $\{Y_1,Y_2,Y_3\}$ is a ${\mathcal G}^{TM}$-orthogonal basis.
Thus $Y_1$ and $Y_2$ are both solutions to Hamilton-Jacobi
differential equation because of
Corollary~\ref{Corol:DecouplingVZero}.

Let us try to find more vector fields solution to Hamilton-Jacobi
differential equation. In order to do this we have to find functions
$f\in {\mathcal C}^\infty(M)$ such that $fY_1$ and $fY_2$ also
satisfy condition~\ref{eq:DecouplingCondition} in
Corollary~\ref{Corol:DecouplingVZero}. According to
Proposition~\ref{prop:Cond-Contr-fX} $f$ must satisfy
\begin{equation*} Y_3(f)=0
\; \Leftrightarrow \;\ds{-\sin \theta \frac{\partial f}{\partial
x}+\cos \theta  \frac{\partial f}{\partial y}+\frac{1}{h}\,
\frac{\partial f}{\partial \theta}=0}.
\end{equation*}
For instance, $f(x,y,\theta)=g(x-h\cos \theta, y-h\sin
\theta)$ where $g\colon \mathbb{R}^2\rightarrow \mathbb{R}$
satisfies the above partial differential equation. Hence,
all vector fields $fY_1$ and $fY_2$ are solutions to
Hamilton-Jacobi differential equation, but not necessarily
their linear combinations.

So far we have found some particular solutions of
Hamilton-Jacobi differential equation. Let us consider now
the most general section in $\Gamma(\tau_D)$, $X=\alpha_1
Y_1+\alpha_2Y_2+\alpha_3Y_3$. Let us find functions
$\alpha_1,\alpha_2,\alpha_3$ such that condition $i_X({\rm
d}^D(\flat_{{\mathcal G}^D}(X)))(Y)= 0$ for all $Y\in
{\mathcal D}_{(c)}^{\perp}$ in Theorem~\ref{maintheorem-5}
is satisfied, that is,
\begin{align*}
0=&-{\mathcal G}^D(Y_1,Y_1)(\alpha_1Y_3(\alpha_1)+ \alpha_2 \alpha_1
\Gamma^1_{32}+\alpha_3\alpha_1\Gamma^1_{33})\\&- {\mathcal
G}^D(Y_2,Y_2)(\alpha_2Y_3(\alpha_2)+\alpha_1\alpha_2\Gamma^2_{31})\\&+{\mathcal
G}^D(Y_3,Y_3)(\alpha_1Y_1(\alpha_3)+\alpha_2Y_2(\alpha_3)+\alpha_1\alpha_2\Gamma^3_{21}).
\end{align*}
It can be proved that this condition and
%For instance $\alpha_3Y_3$ satisfies this condition for any
% $\alpha_3\in {\mathcal C}^\infty(M)$.  Let us verify if
Hamilton-Jacobi differential equation~\eqref{eq:thm5h-j-p}
are satisified, for instance, by
\begin{eqnarray*}
 g(x-h\cos \theta, y-h\sin
\theta)Y_3, && g(x-h\cos \theta, y-h\sin
\theta)Y_2+\alpha_3 Y_3, \\
g(x-h\cos \theta, y-h\sin
\theta)Y_1, && f((x-h\cos \theta)/h)Y_2+\alpha_3 Y_3.
\end{eqnarray*}
%
%
% is satisfied for this vector field
% \begin{equation*}
% {\mathcal
% G}^{TM}(\nabla_{Y_3}\alpha_3Y_3,\alpha_3Y_3)=\alpha_3Y_3(\alpha_3){\mathcal
% G}^{TM}(Y_3,Y_3)=0.
% \end{equation*}
% Equivalently, $Y_3(\alpha_3)=0$. As above, only the
% functions $\alpha_3(x,y,\theta)=g(x-h\cos \theta, y-h\sin
% \theta)$
%  where $g\colon \mathbb{R}^2\rightarrow \mathbb{R}$ make
% $\alpha_3X_3$ solution to Hamilton-Jacobi differential
% equation.

\subsubsection{Robotic leg}\label{Sec:RoboticLeg}

We refer to~\cite[Section 7.4.1]{bullolewis} for a detailed
description of this system. The configuration manifold is
$M=\mathbb{R}^+\times \mathbb{S}^1\times \mathbb{S}^1$ with local
coordinates $(r,\theta,\psi)$. The riemannian metric for the system
is
\[{\mathcal G}^{TM}= m({\rm
d} r \otimes {\rm d}r + r^2{\rm d} \theta \otimes {\rm d}\theta)+J
{\rm d} \psi \otimes {\rm d} \psi,\] where $m$ is the mass of the
particle on the end of the extensible leg and $J$ is the moment of
inertia of the base rigid body about the pivot point.

The control vector fields that span ${\mathcal D}_{(c)}$ are
\begin{equation*}
Y_1=\frac{1}{mr^2} \frac{\partial }{\partial
\theta}-\frac{1}{J}\frac{\partial}{\partial \psi}, \quad
Y_2=\frac{1}{m}\frac{\partial}{\partial r}.
\end{equation*}
There are no constraints on the system, then as in the
previous example $D=TM$, $\lcf \cdot \; , \cdot \rcf_{TM}$
is the usual Lie bracket and $\rho_{TM}={\rm Id}_{TM}$.

The $\mathcal{G}^D$-orthogonal distribution to ${\mathcal D}_{(c)}$
is spanned by
\begin{equation*} Y_3=\frac{\partial }{\partial \psi}
+\frac{\partial}{\partial \theta}.
\end{equation*}
Note that $\{Y_1,Y_2,Y_3\}$ is a $\mathcal{G}^D$-orthogonal basis.
From now on we consider coordinates adapted to this basis. Then the
non-zero Christoffel symbols are:
\begin{equation*}
\begin{array}{lll}
\ds{\Gamma^2_{11}=-\frac{1}{mr^3}}, &
\Gamma^1_{12}=\ds{\frac{J}{mr(J+mr^2)}}, &
\Gamma^3_{12}=\ds{\frac{1}{mr(J+mr^2)}},\\
\Gamma^2_{13}=-\ds{\frac{1}{r}}, &
\Gamma^1_{21}=-\ds{\frac{J}{mr(J+mr^2)}}, &
\Gamma^3_{21}=-\ds{\frac{1}{mr(J+mr^2)}},\\
\Gamma^1_{23}=\ds{\frac{Jr}{J+mr^2}},
&\Gamma^3_{23}=\ds{\frac{r}{J+mr^2}}, &
\Gamma^2_{31}=-\ds{\frac{1}{r}}, \\
\Gamma^1_{32}=\ds{\frac{Jr}{J+mr^2}}, &
\Gamma^3_{32}=\ds{\frac{r}{J+mr^2}}, & \Gamma^2_{33}=-rm.
\end{array}
\end{equation*}
From here, it is easy to compute the following symmetric products:
\begin{equation*}
\langle Y_1\colon Y_1 \rangle_{\mathcal{G}^D}=-\ds{\frac{2}{mr^3}}Y_2, \quad \langle
Y_1\colon Y_2 \rangle_{\mathcal{G}^D}=\langle Y_2\colon Y_2 \rangle_{\mathcal{G}^D}=0.
\end{equation*}
According to Theorem~\ref{thm:MaxRed} the system is maximally
reducible. Thus all the ${\mathcal C}^\infty(M)$-linear combination
of $Y_1$ and $Y_2$ are decoupling sections and we can apply
Corollary~\ref{Corol:DecouplingVZero} to identify those decoupling
sections that are solutions to Hamilton-Jacobi differential equation:
\begin{align*}0&={\mathcal
G}^D(\nabla^{\mathcal{G}^D}_{Y_3}(\alpha_1Y_1+\alpha_2Y_2),
\alpha_1Y_1+\alpha_2Y_2)\\&={\mathcal
G}^D(Y_1,Y_1)\left(\alpha_1Y_3(\alpha_1)+\alpha_1\alpha_2
\Gamma^1_{32}\right)+{\mathcal
G}^D(Y_2,Y_2)\left(\alpha_2Y_3(\alpha_2)+\alpha_1\alpha_2
\Gamma^2_{31}\right).
\end{align*}
It can be proved that if
$\alpha_i(r,\theta,\psi)=f_i(r,\theta-\psi)$ either for $i=1$ or $i=2$
and smooth functions $f_i\colon \mathbb{R}^+\times
\mathbb{S}^1 \rightarrow \mathbb{R}$, then the section
$\alpha_1Y_1+\alpha_2Y_2$ is a solution to Hamilton-Jacobi
differential equation.

By Theorem~\ref{maintheorem-5}, we can also check that
\begin{equation*}
 f(r,\psi-\theta) Y_1 +\alpha_3 Y_3,\quad  f(r,\psi-\theta) Y_3
\end{equation*}
are solutions to Hamilton-Jacobi differential equation.

\subsubsection{Snakeboard}\label{Sec:Snakeboard}

We refer to~\cite[Section 13.4]{bullolewis} and \cite{MuYa}
for a detailed description of this system. The
configuration manifold is $M=SE(2)\times \mathbb{S}^1\times
\mathbb{S}^1$ with local coordinates $(x,y,\theta,
\psi,\phi)$. Consider the following physical parameters:
the mass $m_c$ of coupler, mass $m_r$ of rotor, mass $m_w$
of each wheel assembly, inertia $J_c$ of coupler about
center of mass, inertia $J_r$ of rotor about center of
mass, inertia $J_w$ of wheel assembly about center of mass,
distance $l$ from coupler center of mass to wheel assembly.

\begin{figure}[thpb]
      \centering
\includegraphics[width=6cm]{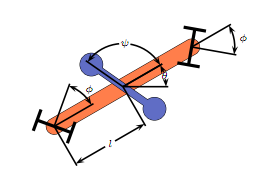}
 \label{Fig-Snakeboard}
\end{figure}

The Riemannian metric in the given coordinates is
\begin{multline*} \mathcal G=(m_c+m_r+2m_w)({\rm d}x\otimes {\rm
d}x+{\rm d}y\otimes {\rm d}y)+(J_c+J_r+2(J_w+m_wl^2)){\rm d}\theta
\otimes {\rm d}\theta \\ + J_r {\rm d}\psi \otimes {\rm d}\psi +2J_w
{\rm d}\phi\otimes {\rm d}\phi+J_r({\rm d}\theta\otimes {\rm
d}\psi+{\rm d}\psi\otimes {\rm d}\theta).
\end{multline*}
The constraints require that the wheels of the snakeboard roll
without slipping. These constraints define a nonholonomic
distribution $D$. As computed in~\cite[Lemma 13.11]{bullolewis} a
$\mathcal G$-orthogonal basis for $D$ on
$M_D=\left\{(x,y,\theta,\psi,\phi)\;| \; \phi \neq \pm
\frac{\pi}{2}\right\}$ is
\begin{align*}
X_1&=\; l\cos \phi \left(\cos \theta \frac{\partial}{\partial
x}+\sin \theta \frac{\partial}{\partial y}\right)-\sin \phi
\frac{\partial}{\partial \theta}=l\cos \phi V_1-\sin \phi
\frac{\partial}{\partial \theta},\\
X_2&=\; a(\phi)V_1-b(\phi)\frac{\partial}{\partial \theta}+
\frac{\partial}{\partial \psi},\\
X_3&=\;\frac{\partial}{\partial \phi},
\end{align*}
where
\begin{equation*}
a(\phi)=\frac{J_rl \cos \phi \sin \phi}{c_1(\phi)},\quad
b(\phi)=\frac{J_r\sin^2\phi}{c_1(\phi)}
\end{equation*}
and
\begin{equation*}
c_1(\phi)=(m_c+m_r+2m_w)l^2\cos^2\phi
+(J_c+J_r+2(J_w+m_wl^2))\sin^2\phi.
\end{equation*}
Let us define a skew-symmetric algebroid structure on the
vector bundle $\tau_D\colon D\rightarrow M_D$. If
$\{e_1,e_2, e_3\}$ is a local basis of $\Gamma(\tau_D)$,
then
\begin{equation*} \rho_D(e_1)=X_1,\quad \rho_D(e_2)=X_2, \quad \rho_D(e_3)=X_3,\label{eq:rhoSnakeboard}\end{equation*}
and
\begin{equation*}
\begin{array}{ll}
\lcf e_1, e_1\rcf_{D}&=\; \lcf e_2, e_2
\rcf_{D}=\; \lcf e_3, e_3 \rcf_{D}=0, \\
 \lcf e_1, e_2\rcf_{D}&=\;0, \\
\lcf e_1, e_3\rcf_{D}&=\;\ds{\frac{{\mathcal
G}([X_1,X_3],X_1)}{{\mathcal G}(X_1,X_1)} \, e_1+
\frac{{\mathcal G}([X_1,X_3],X_2)}{{\mathcal G}(X_2,X_2)} \, e_2},  \\
\lcf e_2, e_3\rcf_{D}&=\;\ds{\frac{{\mathcal
G}([X_2,X_3],X_1)}{{\mathcal G}(X_1,X_1)} \, e_1+ \frac{{\mathcal
G}([X_2,X_3],X_2)}{{\mathcal G}(X_2,X_2)} \, e_2}.
\end{array}\label{eq:bracketSnakeboard}
\end{equation*}
Hence the non-vanishing local structure functions for the
bracket $\lcf \cdot, \cdot \rcf_{D}$ are
\begin{equation*}{\mathcal C}^1_{13},{\mathcal C}^2_{13},{\mathcal
C}^1_{23},{\mathcal C}^2_{23}.
\label{eq:NonLocalStructureFSnakeboard}\end{equation*} The bundle
metric on the skew-symmetric algebroid is given by ${\mathcal
G}^D={\mathcal G}|_{D\times_{M_D} D}$ as explained in
Example~\ref{noholo-example}. We can construct the Levi-Civita
connection $\nabla^{{\mathcal G}^D}$ associated to the bundle metric
${\mathcal G}^D$ having in mind its properties:
\begin{eqnarray*}
 \Gamma^A_{BC}&=&\ds{\frac{1}{2}g^{AA}\left(g_{CC}{\mathcal C}^C_{AB}+g_{AA}{\mathcal C}^A_{BC}
+g_{BB}{\mathcal C}^B_{AC}\right)}\\&+&\ds{\frac{1}{2}g^{AA}\Bigg(-(\rho_D)^i_A\frac{\partial}{\partial x^i}(g_{BC}\delta^{BC})+
(\rho_D)^i_B\frac{\partial}{\partial x^i}(g_{AC}\delta^{AC})}\\&+& \ds{(\rho_D)^i_C\frac{\partial}{\partial x^i}(g_{AB}\delta^{AB})\Bigg)},
\end{eqnarray*}
for a ${\mathcal
G}^D$-orthogonal basis, where $g^{AA}=1/g_{AA}$.

 The non-vanishing Christoffel
symbols are
\begin{equation*} \begin{array}{ll}
\Gamma^1_{13}={\mathcal C}^1_{13}+\frac{1}{2}g^{11}c_1'(\phi), & \Gamma^1_{23}=\frac{1}{2}{\mathcal C}^1_{23}+\frac{1}{2}g^{11}g_{22}
{\mathcal C}^2_{13}, \\  \Gamma^1_{31}=\frac{1}{2}g^{11}c_1'(\phi), &
\Gamma^1_{32}=\frac{1}{2}g^{11}g_{22}{\mathcal C}^2_{13}-\frac{1}{2}{\mathcal C}^1_{23},\\
\Gamma^2_{13}=\frac{1}{2}g^{22}g_{11}{\mathcal C}^1_{23}+\frac{1}{2}{\mathcal C}^2_{13} , &
\Gamma^2_{23}=\frac{1}{2}g^{22}g_{22}'(\phi)+{\mathcal C}^2_{23},\\
 \Gamma^2_{31}=\frac{1}{2}g^{22}g_{11}{\mathcal C}^1_{23}-\frac{1}{2}{\mathcal C}^2_{13},&
 \Gamma^2_{32}=\frac{1}{2}g^{22}g_{22}'(\phi), \\ \Gamma^3_{11}=g^{33}g_{11}{\mathcal C}^1_{31},&
\Gamma^3_{12}=-\frac{1}{2}g^{33}(g_{22}{\mathcal C}^2_{13}+g_{11}{\mathcal C}^1_{23}), \\
\Gamma^3_{21}=-\frac{1}{2}g^{33}(g_{11}{\mathcal C}^1_{23}+g_{22}{\mathcal C}^2_{13}),& \Gamma^3_{22}=\frac{1}{2}g^{33}g_{22}{\mathcal C}^2_{32}.\end{array}
\label{eq:NonZeroChristoffSnakeboard}
\end{equation*}
The mechanical control system on the skew-symmetric algebroid is
given by $(D,{\mathcal G}^D,0,{\mathcal D}_{(c)})$, where
\begin{equation*}
{\mathcal D}_{(c)}=\left\{\ds{\frac{c_1(\phi)}{J_rc_2(\phi)}\, X_2,
\; \frac{1}{2J_w} \, X_3}\right\},
\end{equation*}
where $c_2(\phi)=(m_c+m_r+2m_\omega)l^2\cos^2 \phi + (J_c+2(J_w+m_\omega l^2))\sin^2\phi$.

The control forces are torques, one actuating the rotor and the
other one actuating the wheels. It is easy to prove that
\begin{equation*}
{\mathcal D}_{(c)}^\perp=\left\{X_1\right\}.
\end{equation*}
In order to find solutions to Hamilton-Jacobi differential equation,
let us first check if there exist any decoupling section of
$\Gamma({\mathcal D}_{(c)})$ so that we can use
Corollary~\ref{Corol:DecouplingVZero}.
\begin{align*}
\langle X_2\colon X_2 \rangle_{{\mathcal G}^D}&=\; 2\Gamma^3_{22}
X_3 \in {\mathcal D}_{(c)},\\
\langle X_3\colon X_3 \rangle_{{\mathcal G}^D}&=0 \in {\mathcal
D}_{(c)}.
\end{align*}
Hence both control sections are decoupling. As in the previous
example the mechanical control system is not maximally reducible to
a driftless system because
\begin{equation*}
\langle X_2\colon X_3 \rangle_{{\mathcal G}^D}=\;
(\Gamma^1_{23}+\Gamma^1_{32})X_1+ (\Gamma^2_{23}+\Gamma^2_{32})X_2
 \notin {\mathcal D}_{(c)}.
\end{equation*}

Let us check if the decoupling sections satisfy
condition~\ref{eq:DecouplingCondition} in
Corollary~\ref{Corol:DecouplingVZero}.
\begin{align*}
{\mathcal G}^D(\nabla^{{\mathcal G}^D}_{X_1}X_2,X_2)&=\;
\Gamma^3_{12}{\mathcal G}^D(X_3,X_2)=0,\\
{\mathcal G}^D(\nabla^{{\mathcal G}^D}_{X_1}X_3,X_3)&=\;
\Gamma^1_{13}{\mathcal G}^D(X_1,X_3)+\Gamma^2_{13}{\mathcal
G}^D(X_2,X_3)=0
\end{align*}
because $\{X_1,X_2,X_3\}$ is ${\mathcal G}^D$-orthogonal basis. Thus
$X_2$ and $X_3$ are both solutions to Hamilton-Jacobi differential
equation because of Corollary~\ref{Corol:DecouplingVZero}. Let us see if these solutions can generate
more decoupling sections being solution to  Hamilton-Jacobi differential
equation. As proved in Proposition~\ref{prop:Cond-Contr-fX}, $fX_2$ and $fX_3$ with $f\in {\mathcal C}^\infty(M)$ are solutions to
Hamilton-Jacobi differential
equation if and only if
\begin{equation*}
X_1(f)=0 \; \Leftrightarrow \; \ds{l\cos \phi \left(\cos \theta
\frac{\partial f}{\partial x}+\sin \theta \frac{\partial f}{\partial
y}\right)-\sin \phi \frac{\partial f}{\partial \theta}}=0.
\end{equation*}
If we take any function \begin{equation*}f(x,y,\theta,\psi,\phi)=g\left(\frac{2x\sin \phi}{l}+2 \cos \phi \sin \theta,
y-l\frac{\cos \phi}{\sin \phi} \cos \theta, \psi,\phi\right),\end{equation*} where $g\colon \mathbb{R}^4 \rightarrow \mathbb{R}$
is a solution to the
partial differential equation $X_1(f)=0$, then $fX_2$ and
$fX_3$ are also solutions to Hamilton-Jacobi differential
equation.

Consider now a general section $\alpha_1X_1+\alpha_2X_2+\alpha_3X_3$ in $\Gamma(\tau_D)$. Condition $i_X{\rm
d}^D(\flat_{{\mathcal G}^D}(X))(Y)= 0$ for all $Y\in {\mathcal
D}_{(c)}^{\perp}$ in Theorem~\ref{maintheorem-5} becomes
\begin{align*} {\mathcal
G}^D(X_1,X_1)&(\alpha_2X_2(\alpha_1)+\alpha_3X_3(\alpha_1)+\alpha_3
\alpha_1
\Gamma^1_{31}+\alpha_2\alpha_3\Gamma^1_{23}+\alpha_3\alpha_2\Gamma^1_{32})\\&-
{\mathcal
G}^D(X_2,X_2)(\alpha_2X_1(\alpha_2)+\alpha_3\alpha_2\Gamma^2_{13})\\
&-{\mathcal
G}^D(X_3,X_3)(\alpha_3X_1(\alpha_3)+\alpha_1\alpha_3\Gamma^3_{11}+\alpha_3\alpha_2\Gamma^3_{12})=0.
\end{align*}
For instance, for any $\alpha_1 \in {\mathcal C}^\infty(M)$ $\alpha_1X_1$ satisfies this condition. It will be a solution to
Hamilton-Jacobi differential equation if it
fulfills~\eqref{eq:thm5h-j-p}, that is,
\begin{align*}
0=&{\mathcal G}^D(\nabla^{{\mathcal
G}^D}_{\alpha_1X_1}\alpha_1X_1,X_1)=\alpha_1X_1(\alpha_1){\mathcal
G}^D(X_1,X_1)+\alpha_1^2\Gamma^3_{11}{\mathcal
G}^D(X_3,X_1)\\=&\alpha_1X_1(\alpha_1){\mathcal G}^D(X_1,X_1).
\end{align*}
Equivalently, $X_1(\alpha_1)=0$. Thus $\alpha_1X_1$ is also a solution to Hamilton-Jacobi differential
equation for any \begin{equation*}\alpha_1(x,y,\theta,\psi,\phi)=g\left(\frac{2x\sin \phi}{l}+2 \cos \phi \sin \theta,
y-l\frac{\cos \phi}{\sin \phi} \cos \theta, \psi,\phi\right)\end{equation*} where $g\colon \mathbb{R}^4 \rightarrow \mathbb{R}$.

%%%%%%%%%%%%%%%%%%%%%%%%%%%%%%%%%%%%%%%%%%%%%%%%%%%%%%%%%%%%%%%%%%%%%%%%%%%%%%%%%%%%%%%%%%%%%%%%%%%%%%%%%%%%%%%%%%%%
%\begin{remark}
%Is the Hamilton-Jacobi solution obtained for the snakeboard in \cite{LeMaMa} solution to Hamilton-Jacobi differential
%equation here? Be careful of the coordinates, they are not the same. The answer might be no since in this paper $\flat_{{\mathcal G}^D}(X)$
%is not necessarily a 1-cocycle.
%\end{remark}
%%%%%%%%%%%%%%%%%%%%%%%%%%%%%%%%%%%%%%%%%%%%%%%%%%%%%%%%%%%%%%%%%%%%%%%%%%%%%%%%%%%%%%%%%%%%%%%%%%%%%%%%%%%%%%%%%%%%

\section{Future work}

The results in this paper extend the notion of decoupling sections for mechanical systems with nonzero potential. A decoupling
section for those mechanical systems is a section that satisfies the assumption and condition (ii) in Theorem~\ref{maintheorem-5}.
The future research line consists of taking advatange of this geometric description of decoupling sections to do motion planning
for mechanical systems with nonzero potential. One of the key points to succeed in motion planning is that not any
reparametrization of decoupling sections in the sense of Theorem~\ref{maintheorem-5} is again decoupling in the sense of
 Theorem~\ref{maintheorem-5}. In order to define the suitable notion of reparametrization
Proposition~\ref{prop:Cond-Contr-fX} seems to be useful.

 \end{document}